\documentclass[11pt,letter, leqno]{article}
\usepackage[margin=1in]{geometry}
\usepackage[utf8]{inputenc}
\usepackage{amssymb}
\usepackage{bbm}	
\usepackage{amsthm,amssymb,amsmath}
\usepackage{mathtools}
\usepackage{thmtools,thm-restate}
\usepackage{tikz}
\usetikzlibrary{arrows,automata}
\usepackage{multirow}
\usepackage{tabularx}
\usepackage[ruled,lined,boxed]{algorithm2e}
\usepackage{algorithmic}
\usepackage{comment}
\usepackage[caption=false]{subfig}
\usepackage{booktabs}
\usepackage{libertine}
\usepackage{authblk}
\usepackage{optidef}
\usepackage{stackengine}
\usepackage{amsmath,xparse}
\usepackage{array}
\newcolumntype{P}[1]{>{\centering\arraybackslash}p{#1}}
\newcolumntype{M}[1]{>{\centering\arraybackslash}m{#1}}

\newdimen\prevdp
\def\leftlabel#1{\noalign{\prevdp=\prevdepth
   \kern-\prevdp\nointerlineskip\vbox to0pt{\vss\hbox{#1}}\kern\prevdp}}

\usepackage{thm-restate}
\usepackage[libertine]{newtxmath} 
\usepackage[scaled=0.96]{zi4} 

\usepackage{ dsfont, microtype, xcolor, paralist}

\usepackage[ocgcolorlinks]{hyperref} 
\colorlet{DarkRed}{red!50!black}
\colorlet{DarkGreen}{green!50!black}
\colorlet{DarkBlue}{blue!50!black}
\hypersetup{
	linkcolor = DarkRed,
	citecolor = DarkGreen,
	urlcolor = DarkBlue,
	bookmarksnumbered = true,
	linktocpage = true
}

\newcommand{\ED}[1]{\textcolor{blue}{[ED:#1]}}
\usepackage{xspace}
\newcommand{\UBudgetName}{\textbf{B-URST}\xspace}
\newcommand{\UQuotaName}{\textbf{Q-URST}\xspace}
\newcommand{\DSteinerT}{\textbf{DSteinerT}\xspace}
\newcommand{\UndirectedBudgetAdditive}{\textbf{B-URAT}\xspace}

\newcommand{\DBudgetAdditiveName}{\textbf{B-DRAT}\xspace}
\newcommand{\DQuotaAdditiveName}{\textbf{Q-DRAT}\xspace}

\newcommand{\PrizeColectingSteinerTree}{\textbf{PCST}\xspace}

\newcommand{\DBudgetAdditiveLP}{\textbf{\DBudgetAdditiveName-LP}\xspace}
\newcommand{\DQuotaAdditiveLP}{\textbf{\DQuotaAdditiveName-LP}\xspace}

\newcommand{\UBudgetLP}{\textbf{\UBudgetName-LP}\xspace}
\newcommand{\UQuotaLP}{\textbf{\UQuotaName-LP}\xspace}

\let\epsilon\varepsilon
\let\eps\varepsilon

\usepackage{xspace}

\usepackage{tikz}
\usetikzlibrary{decorations.pathreplacing}
\usetikzlibrary{plotmarks}
\usetikzlibrary{positioning,automata,arrows}
\usetikzlibrary{shapes.geometric}
\usetikzlibrary{decorations.markings}
\usetikzlibrary{positioning, shapes, arrows}

\PassOptionsToPackage{usenames,dvipsnames,svgnames}{xcolor}
\definecolor{orange}{RGB}{235,90,0}
\definecolor{darkorange}{RGB}{175,30,0}
\definecolor{turkis}{RGB}{131,182,182}
\definecolor{darkturkis}{RGB}{31,82,82}
\definecolor{green}{RGB}{102,180,0}
\definecolor{darkgreen}{RGB}{51,90,0}
\definecolor{myblue}{RGB}{0,0,213}
\definecolor{mydarkblue}{RGB}{0,0,100}
\definecolor{mybrightblue}{HTML}{74B0E4}
\definecolor{mybrighterblue}{HTML}{B3EAFA}
\definecolor{lila}{RGB}{102,0,102}
\definecolor{darkred}{RGB}{139,0,0}
\definecolor{darkyellow}{RGB}{188,135,2}
\definecolor{brightgray}{RGB}{200,200,200}
\definecolor{darkgray}{RGB}{50,50,50}
\definecolor{amaranth}{rgb}{0.9, 0.17, 0.31}
\definecolor{alizarin}{rgb}{0.82, 0.1, 0.26}
\definecolor{amber}{rgb}{1.0, 0.75, 0.0}
\definecolor{green(ryb)}{rgb}{0.4, 0.69, 0.2}
\definecolor{hanblue}{rgb}{0.27, 0.42, 0.81}
\definecolor{grannysmithapple}{rgb}{0.66, 0.89, 0.63}

\usepackage{amsthm}

\newtheorem{theorem}{Theorem}[section]
\newtheorem{lemma}{Lemma}[section]

\newtheorem{claim}{Claim}[section]

\newtheorem{lemma-rstbl}{Lemma}[section]
\newtheorem{obs-rstbl}{Observation}[section]
\newtheorem{theorem-rstbl}{Theorem}[section]

\usepackage{environ}
\NewEnviron{cproblem}[1]{%
\begin{center}\fbox{\parbox{5.5in}{%
    {\centering\scshape #1\par}%
    \parskip=1ex
    \everypar{\hangindent=1em}%
    \BODY
}}\end{center}}

\title{
Approximation algorithms for Node-weighted Steiner Problems: \\Digraphs with Additive Prizes and Graphs with Submodular Prizes
}
\author{Gianlorenzo D'Angelo}
\author{Esmaeil Delfaraz}
\affil{\normalsize Gran Sasso Science Institute, L'Aquila, Italy}
\date{}

\begin{document}

\maketitle
\begin{abstract}
\begin{sloppypar}
In the \emph{budgeted rooted node-weighted Steiner tree} problem, we are given a graph $G$ with $n$ nodes, a predefined node $r$, two weights associated to each node modelling costs and prizes, and a real-valued budget $B$. The aim is to find a tree in $G$ rooted at $r$ such that the total cost of its nodes is at most $B$ and the total prize is maximized. In the \emph{quota rooted node-weighted Steiner tree} problem, we are given a real-valued quota $Q$, instead of the budget, and we aim at minimizing the cost of a tree rooted at $r$ whose overall prize is at least $Q$.

In this paper, we study some relevant generalizations of the above problems, namely we consider separately (i) directed graphs (with additive prize function), and (ii) monotone submodular prize functions over subsets of nodes (in undirected graphs).


For the budgeted problem, Ghuge and Nagarajan~[SODA\ 2020], and D'Angelo, Delfaraz and Gilbert~[ISAAC\ 2022] proposed an optimal quasi-polynomial time $O\left(\frac{\log n'}{\log \log n'}\right)$-approximation algorithm and a polynomial time bicriteria $(1+\epsilon, O(\frac{1}{\epsilon^{3}}\sqrt{B}))$-approximation algorithm, respectively, for the case in which the graph is directed, the costs are restricted to be positive integer valued and are associated to the edges, and the prize function is a general monotone submodular function over subsets of nodes, where $n'$ is the number of vertices in an optimal solution and $\epsilon \in (0, 1]$.

For \emph{scenario (i)}, we develop a technique resorting on a standard flow-based linear programming relaxation to compute a tree with good trade-off between prize and cost, which allows us to provide very simple polynomial time approximation algorithms for both the budgeted and the quota problems. For the \emph{budgeted} problem, our algorithm achieves an $O(\frac{1}{\epsilon^2}n^{2/3}\ln{n})$-approximation at the cost of a budget violation of a factor of at most $1+\epsilon$, for any $\epsilon \in (0, 1]$. For the \emph{quota} problem, our algorithm guarantees an approximation factor of $O(n^{2/3}\ln{n})$ at the cost of a violation of the quota constraint by a factor of at most $2$. Our algorithms work for nonnegative real-valued costs and are the first non-trivial polynomial time algorithms with approximation factors depending on $n$.

The same technique can be used to find an approximation algorithm for the node-weighted directed Steiner tree problem (\DSteinerT).
Recently, Li and Laekhanukit~[SODA\ 2022] ruled out poly-logarithmic approximation algorithms for \DSteinerT based on a standard flow-based linear programming relaxation.
%
%
By using the same linear relaxation, we provide a surprisingly simple polynomial time $O((1+\epsilon)\sqrt{n} \ln {n})$-approximation algorithm for \DSteinerT, for any $\epsilon>0$. 

For \emph{scenario (ii)}, we provide a polynomial time $O(\frac{1}{\epsilon^3}\sqrt{n}\log{n})$-approximation algorithm for the budgeted problem that violates the budget constraint by a factor of at most $1+\epsilon$, for any $\epsilon \in (0, 1]$. Also in this case we introduce a general technique that exploits a flow-based linear program to find trees with a good trade-off between prize and cost, which allows us to provide a good approximation  also for the quota problem.
\end{sloppypar}
\end{abstract}


\section{Introduction}\label{sec:intro}

\emph{Prize Collecting Steiner Tree} (\PrizeColectingSteinerTree) refers to a wide class of combinatorial optimization problems, involving variants of the \emph{Steiner tree} and \emph{traveling sales man} problems, with many practical applications in computer and telecommunication networks, VLSI design, computational geometry, wireless mesh networks, and cancer genome studies~\cite{cheng2004steiner, gao2018algorithm,hochbaum2020approximation, kuo2014maximizing,vandin2011algorithms}. 


In general, we are given a (directed) graph $G$, two functions modelling costs and prizes (or penalties) associated to the edges and/or to the nodes of the graph, and we want to find a connected subgraph of $G$ (usually a tree or an out-tree) $T$ which optimizes an objective function defined as a combination of its costs and prizes and/or is subject to some constraints on its  cost and prize. Casting suitable constraints and objective functions give rise to different problems.
For example, in \emph{budgeted} problems, we are given a budget $B$ and we require that the cost of $T$ is no more than $B$ and its prize is maximized. In \emph{quota} problems, we require the prize of $T$ to be at least some quota $Q$ and its cost to be minimum.
Additional constraints can be required, for example in \emph{rooted} variants we are given a specific node, called \emph{root}, which is required to be part of $T$ and reach all the nodes in $T$, while in \emph{Steiner tree} problems $T$ must include a specified set of nodes called \emph{terminals}.

While there is a vast literature providing
approximation algorithms for many variants of \PrizeColectingSteinerTree on undirected graphs where the prize function is additive, e.g.~\cite{archer2011improved, bateni2018improved, garg2005saving,  guha1999efficient, goemans1995general, johnson2000prize, konemann2013lmp, paul2020budgeted}, the case of directed graphs or monotone  submodular prize functions received less attention~\cite{charikar1999approximation, d2022budgeted, ghuge2020quasi, kuo2014maximizing, zelikovsky1997series}.

In this paper, we consider \emph{node-weighted} Steiner tree problems, that is both costs and prizes  are associated to the nodes of the graph, and we investigate two relevant settings: (i) the underlying graph is \emph{directed} and the prize function is \emph{additive}, and (ii) the underlying graph is \emph{undirected} and the prize function is \emph{monotone and submodular}.
In both settings, we consider budgeted and quota problems. For the first setting we also study the minimum-cost Steiner tree problem. We consider the more general rooted variant of all these problems. For each of the above two settings, we introduce a new technique, resorting on flow-based linear programming relaxations, which allows us to find trees or out-trees with a good trade-off between cost and prize. Casting suitable values of quota and budget and applying new and known tree trimming procedures, we can achieve good bicriteria approximation ratios for all the above problems.

In the following we analyze  settings (i) and (ii) separately.


\subsection{Directed Graphs and Additive Prizes}

Prize collecting Steiner trees are usually much harder on directed graphs than on undirected graphs. A well-known example is the \emph{Steiner tree} problem, for which there is a simple polynomial time $2$-approximation algorithm for its undirected version, but for its directed version, unless $NP \subseteq \bigcap_{0 <\epsilon< 1} \text{ZPTIME}(2^{n^\epsilon})$ or the Projection Game Conjecture is false, there is no quasi-polynomial time algorithm that achieves an approximation ratio of $o(\frac{\log^2 k}{\log \log k})$~\cite{grandoni2019log2}, where $k$ is the number of terminal nodes.


We consider three variants of the node-weighted version of \PrizeColectingSteinerTree in this scenario in which we are given a directed graph $D=(V, A)$ with $|V|=n$, two nonnegative functions, namely, cost $c(v)$ and prize $p(v)$ that are associated to each vertex $v \in V$, and a root vertex $r$.

\paragraph{The Budgeted Rooted Additive Tree problem (\DBudgetAdditiveName).} Given a budget $B$, in \DBudgetAdditiveName, we aim to find an out-tree (a.k.a. out-arborescence) $T$ of $D$ rooted at $r$ such that the sum of the costs of its vertices is at most $B$ and the sum of the prizes of its vertices is maximum.

\cite{guha1999efficient} introduced the undirected version of \DBudgetAdditiveName, we call this problem \UndirectedBudgetAdditive. They gave a polynomial time $O(\log^2 n)$-approximation algorithm that violates the budget constraint by a factor of at most $2$. The approximation factor was improved to $O(\log n)$, with the same budget violation, by Moss and Rabani~\cite{moss2007approximation}. Bateni, Hajiaghay and Liaghat~\cite{bateni2018improved} later improved the budget violation factor to $1+\epsilon$ to obtain an approximation factor of $O\left(\frac{1}{\epsilon^2}\log n\right)$, for any $\epsilon\in (0,1]$.
Kortsarz and Nutov~\cite{kortsarz2009approximating} showed that the unrooted version of \UndirectedBudgetAdditive, so does \UndirectedBudgetAdditive, admits no $o(\log \log n)$-approximation algorithm, unless $NP \subseteq DTIME(n ^{\text{polylog}(n)})$, even if the algorithm is allowed to violate the budget constraint by a factor equal to a universal constant. 

Ghuge and Nagarajan~\cite{ghuge2020quasi} provided a tight quasi-polynomial time $O(\frac{\log n'}{\log \log n'})$-approximation algorithm for the edge-cost version of \DBudgetAdditiveName that requires $(n \log B)^{O(\log^{1+\epsilon} n')}$ time, where $n'$ is the number of vertices in an optimal solution and the edge costs are positive integers. Very recently, D'Angelo, Delfaraz and Gilbert~\cite{d2022budgeted} provided a polynomial time $O\left(\frac{1}{\epsilon^2}\sqrt{B}\right)$-approximation algorithm for the same problem that violates the budget constraint by a factor of $1+\epsilon$, where $\epsilon \in (0, 1]$. Bateni, Hajiaghay and Liaghat~\cite{bateni2018improved} showed that the integrally gap of the flow-based LP for \UndirectedBudgetAdditive, so is for \DBudgetAdditiveName, is unbounded. 

In this paper we show that, by using this flow-based LP and violating the budget constraint by a factor of at most $1+\epsilon$, for $\epsilon \in (0, 1]$, one can achieve a very simple polynomial time $O\left(\frac{1}{\epsilon^2}n^{2/3} \ln{n}\right)$-approximation algorithm for \DBudgetAdditiveName where costs and prizes are non-negative real numbers.

\paragraph{The Quota Rooted Additive Tree problem (\DQuotaAdditiveName).} Given a quota $Q$, in \DQuotaAdditiveName, we aim to find an out-tree $T$ of $D$ rooted at $r$ whose total prize is at least $Q$ and total cost is minimum. To the best of
our knowledge, this problem has not been studied explicitly before.

On undirected graphs for \DQuotaAdditiveName, which is called \textbf{Q-URAT}, by using the algorithm by Moss and Rabbani~\cite{moss2007approximation} along with the ideas of K{\"o}nemann and Sadeghian and Sanit{\`{a}}~\cite{konemann2013lmp}, and Bateni, Hajiaghay and Liaghat~\cite{bateni2018improved}, one can provide an $O(\log{n})$-approximation algorithm. 

Here, by using the flow-based LP, we present a very simple $O((1+\epsilon)n^{2/3} \ln {n})$-approximation algorithm for \DQuotaAdditiveName that violates the quota constraint by a factor of at most $2$ and runs in a time that is polynomial in the input size and in $1/\epsilon$, for any $\epsilon >0$.

\paragraph{The Node-Weighted Directed Steiner Tree problem (\DSteinerT).} In \DSteinerT, we are given a set of terminals $K \subseteq V$ and the goal is to find an out-tree $T$ of $D$ rooted at $r$ spanning the set $K$ whose total cost of vertices is minimum.

Zelikovsky~\cite{zelikovsky1997series} provided the first approximation algorithm for the edge-cost version of \DSteinerT, we term this problem \textbf{E-}\DSteinerT, with factor $O(k^{\epsilon} (\log^{1/\epsilon}{k}))$ that runs in $O(n^{1/\epsilon})$, where $|K|=k$. A follow-up work by~\cite{charikar1999approximation} proposed a better approximation algorithm with a factor $O(\log^3{k})$ in quasi-polynomial time. This factor was improved to $O(\frac{\log^2 k}{\log \log k})$ by the randomized algorithm of Grandoni, Laekhanukit, and Li.~\cite{grandoni2019log2} in $n^{O(\log^{5} k)}$ time. 
Later, Ghuge and Nagarajan~\cite{ghuge2020quasi} proposed a deterministic $O(\frac{\log^2 k}{\log \log k})$-approximation algorithm for \textbf{E-}\DSteinerT in $n^{O(\log^{1+\epsilon} k)}$ time. Recently, Li and Laekhanukit~\cite{li2022polynomial} ruled out poly-logarithmic approximation algorithms for \textbf{E-}\DSteinerT using the standard flow-based LP. This result holds for \DSteinerT too as, in the instance in~\cite{li2022polynomial}, the incoming edges of each vertex have the same cost.

Using this standard flow-based LP, in this work, we show that one can achieve a very simple $O((1+\epsilon)\sqrt{n} \ln {n})$-approximation algorithm for \DSteinerT, that requires a time that is polynomial in the input size and in $1/\epsilon$, for any $\epsilon >0$.

Our algorithm does not improve over the best known algorithms in terms of approximation ratio, however it is much simpler and provides a guarantee over the optimum of the standard flow-based linear programming relaxation, which allows us to use it as a black-box to approximate other problems. As an example, we show that it can be used to approximate \DBudgetAdditiveName and \DQuotaAdditiveName with factors 
$O(\frac{n^{3/4}\ln{n}}{\epsilon^2})$ and
$O(n^{3/4}\ln{n})$ and constraint violation of $1+\epsilon$ and 2, respectively, for any $\epsilon\in (0,1]$.

\subsection{Undirected Graphs and Submodular Prizes}
There are only a few papers that provide approximation algorithms for some variants of \PrizeColectingSteinerTree on undirected graphs when the prize function is monotone submodular~\cite{d2022budgeted, HajiaghayiKKN12, khuller2020analyzing, kuo2014maximizing, lamprou2020improved, ran2016approximation}.

We consider two variants of node-weighted \PrizeColectingSteinerTree in this scenario.
Let $G=(V, E)$ be an undirected graph with $n$ nodes, $c:V \rightarrow \mathbb{R}^{\ge 0}$ be a cost function on nodes, $p:2^V \rightarrow \mathbb{R}^{\ge 0}$ be a monotone submodular prize function on subsets of nodes and $r\in V$ be a root vertex. 

\paragraph{The Budgeted Undirected Rooted Submodular problem (\UBudgetName).} Given a budget $B \in \mathbb{R}^{+}$, in \UBudgetName we aim to find a tree $T$ of $G$ including $r$ that costs at most $B$ and maximizes the prize function.

Kuo, Lin, and Tsai~\cite{kuo2014maximizing} studied the unrooted version of \UBudgetName, when all node costs are positive integers. They provided an $O(\Delta \sqrt{B})$-approximation algorithm, where $\Delta$ is the maximum degree of the graph.  Very recently, this factor was improved to $O(\sqrt{B})$ by D'Angelo, Delfaraz and Gilbert~\cite{d2022budgeted}. The authors in~\cite{d2022budgeted} also provided a bicriteria $O(1+\epsilon, \sqrt{B}/\epsilon^3)$ for \UBudgetName (i.e. in the rooted case) when all node costs are positive integers. Ghuge and Nagarajan~\cite{ghuge2020quasi} provided a tight quasi-polynomial time $O(\frac{\log n'}{\log \log n'})$-approximation algorithm for \UBudgetName on directed graphs that requires $(n \log B)^{O(\log^{1+\epsilon} n')}$ time, where $n'$ is the number of vertices in an optimal solution and the edge costs are positive integers.

In this paper, we provide a new flow-based linear programming formulation that leads to a polynomial time $O(\frac{1}{\epsilon^3}\sqrt{n}\log{n})$-approximation algorithm for \UBudgetName violating the budget constraint by a factor of at most $1+\epsilon$, for any $\epsilon \in (0, 1]$.
    
\paragraph{The Quota Undirected Rooted Submodular problem (\UQuotaName).} Given a quota $Q \in \mathbb{R}^{+}$, in \UQuotaName, we aim to find a tree $T$ of $G$ including $r$ such that the prize function of $T$ is at least $Q$ and the cost of $T$ is minimum. To the best of our knowledge, this problem has not been studied explicitly before.

We present an algorithm that finds trees with a good trade-off between prize and cost, which allows us to provide a good approximation for \UQuotaName.

\subsection{Summary of the Results}

In the first step of our study, we present a surprisingly simple technique that yields a polynomial time $O((1+\epsilon)\sqrt{n}\ln{n})$-approximation algorithm for \DSteinerT, for any $\epsilon >0$. Then we show that using the core idea of our approach for \DSteinerT, one can achieve very simple approximation algorithms for \DQuotaAdditiveName and \DBudgetAdditiveName. Our algorithm for \DQuotaAdditiveName guarantees an $O((1+\epsilon)n^{2/3}\ln{n})$-approximation that violates the quota constraint by a factor of at most $2$, for any $\epsilon >0$. Our algorithm for \DBudgetAdditiveName achieves an approximation factor of $O(\frac{1}{\epsilon^{2}}n^{2/3}\ln{n})$ with $1+\epsilon$ budget violation, for any $\epsilon \in (0, 1]$.  

In the next step, we introduce a novel LP relaxation for \UBudgetName and \UQuotaName. By using it, we show that one can achieve a polynomial time $O(\frac{1}{\epsilon^{3}}\sqrt{n}\log{n})$-approximation algorithm for \UBudgetName, violating the budget constraint by a factor of at most $1+\epsilon$, for any $\epsilon \in (0, 1]$. Last, by using our LP, we show that for \UQuotaName, our technique achieves either (i) an $O(\sqrt{n} \log{n})$-approximation that violates the quota constraint by a factor of at most $2$, or (ii) a $(1+\epsilon)$-approximation algorithm that violates the quota constraint by a factor of at most $2\sqrt{n}$, for any $\epsilon >0$.

Table~\ref{tbBestBounds} summarizes the best previous results on our problems and our results.
\begin{table}[t]
\caption{A summary of the best previous results and our results on our problems. $\epsilon \in (0, 1]$, $k$ is the number of terminals and $n'$ is the number of vertices in an optimal solution. For any $\alpha, \beta \ge 1$, $(\beta, \alpha)$ represents an $\alpha$-approximation with a $\beta$ violation in the constraint for the corresponding problem. Note that all the previous best results are for the directed edge-cost versions and (except~\cite{zelikovsky1997series}) are only for the case when the costs are positive integer.}\label{tbBestBounds}
\begin{center}
\begin{tabular}{|M{1.7cm}|M{6cm}|M{5cm}|}
 \hline
 Problem & Best Previous Results & This Paper \\ \hline
 \DBudgetAdditiveName & $O(\frac{\log{n'}}{\log{\log{n'}}})$~\cite{ghuge2020quasi} (quasi-poly-time), $(1+\epsilon, O(\frac{\sqrt{B}}{\epsilon^2}))$~\cite{d2022budgeted} (poly-time) & 
 $(1+\epsilon, O(\frac{n^{2/3} \ln {n}}{\epsilon^2}))$ (Theorem~\ref{thMainDBudget})\\ \hline
 \DQuotaAdditiveName & - & $(2, O(n^{2/3}\ln {n}))$ (Theorem~\ref{thDQuotaTree})\\ \hline
 \DSteinerT & $O(\frac{\log^2{k}}{\log{\log{k}}})$~\cite{ghuge2020quasi} (quasi-poly-time), $ O(k^{\epsilon}\log^{1/\epsilon}{k})$~\cite{zelikovsky1997series} (poly-time) & $O((1+\epsilon)\sqrt{n}\ln{n})$ (Theorem~\ref{thDSteinerTree})\\ \hline
 \UBudgetName & $O(\frac{\log{n'}}{\log{\log{n'}}})$~\cite{ghuge2020quasi} (quasi-poly-time), $(1+\epsilon, O(\frac{\sqrt{B}}{\epsilon^3}))$~\cite{d2022budgeted} (poly-time) & $(1+\epsilon, O(\frac{\sqrt{n}\log {n}}{\epsilon^3}))$(Theorem~\ref{thMainUBudget})\\ \hline
 \UQuotaName & - & $(2, O(\sqrt{n} \log{n}))$ or $(2\sqrt{n}, 1+\epsilon)$ (Theorem~\ref{thUQuotaTree})\\ \hline
\end{tabular}
\end{center}
\end{table}

\subsection{Further Related Work}
Here we describe those variants of prize collecting Steiner tree problems that are more closely related to our study.

\cite{charikar1999approximation} provided a polynomial time $\tilde O(k^{2/3})$-approximation algorithm for $k$-Directed Steiner Forest ($k-$\textbf{DSF}) which is a more general version of the directed Steiner problem, where the goal is to find a minimum cost sugraph that connects at least $k$ predetermined pairs. Feldman, Kortsarz, and Nutov~\cite{feldman2012improved} improved this bound to $O(k^{1/2+\epsilon})$ for any fixed $\epsilon >0$. A well-known variant of $k-$\textbf{DSF} is \textbf{DSF} in which the goal is to find a minimum cost subgraph that contains a path between every prescribed pair. \cite{chekuri2011set} provided an $O(k^{1/2+\epsilon})$-approximation algorithm for \textbf{DSF}. In terms of the number of vertices $n$, Feldman, Kortsarz, and Nutov~\cite{feldman2012improved} gave an $O(n^{\epsilon}\cdot \min\{n^{4/5}, m^{2/3}\})$-approximation algorithm for \textbf{DSF} which was improved by~\cite{berman2011improved} to $O(n^{2/3+\epsilon})$, for any $\epsilon>0$. Later,~\cite{chlamtavc2020approximating} proposed an $O(n^{3/5+\epsilon})$-approximation algorithm for \textbf{DSF} when all edges in a given directed graph have the same cost.

Hochbaum and Rao~\cite{hochbaum2020approximation} investigated \UBudgetName in which each vertex costs $1$ and provided an approximation algorithm with factor $\min\{1/((1-1/e)(1/R-1/B)), B\}$, where $R$ is the radius of the graph. \cite{chen2020optimal} investigated the edge-cost version of \UBudgetName. The authors in~\cite{vandin2011algorithms} proposed a $(\frac{2e-1}{e-1}R)$-approximation algorithm for a special case of \UBudgetName, where $R$ is the radius of an optimal solution. This problem coincides with the connected maximum coverage problem in which each set costs one. In the connected maximum coverage problem, we are given a graph in which each vertex represents a subset of elements and is associated with some nonnative cost, and each element is associated with some prize. The goal here is to find a connected subgraph whose total cost is bounded by a given budget and prize is maximum. \cite{ran2016approximation} presented an $O(\Delta\log{n})$-approximation algorithm for a special case of the connected maximum coverage problem, where $\Delta$ is the maximum degree in the graph. This bound was recently improved to $O(\log{n})$ by D'Angelo, Delfaraz and Gilbert~\cite{d2022budgeted}. The authors in \cite{d2022budgeted} also showed that one can obtain an $O(f\log{n})$-approximation algorithm for the connected maximum coverage problem, where $f$ is the maximum frequency of an element. 

One of the applications of \UBudgetName is a problem in wireless sensor networks called \emph{Budgeted Sensor Cover Problem} (\textbf{BSCP}), where the goal is to find a set of connected sensors of size at most $B$ to maximize the number of covered users. The authors in~\cite{kuo2014maximizing} provided a $5(\sqrt{B}+1)/(1-1/e)$-approximation algorithm for \textbf{BSCP}, which was improved by~\cite{xu2021throughput} to $\lfloor\sqrt{B}\rfloor/(1-1/e)$. Huang, Li, and Shi~\cite{huang2020approximation} proposed a $8(\lceil 2\sqrt{2} C\rceil+1)^2/(1-1/e)$-approximation algorithm for \textbf{BSCP}, where $C=O(1)$. Danilchenko, Segal, and Nutov~\cite{danilchenko2020covering} investigated a closely related problem to \textbf{BSCP}, where the goal is to place a set of connected disks (or squares) in the plane such that the total weight of target points is maximized. They proposed a polynomial time $O(1)$-approximation algorithm for this problem. Another variant of \UBudgetName is the budgeted connected dominating set problem in which we need to find a connected subset of vertices of size at most $B$ in a given undirected graph that maximizes the profit function on the dominated vertices. Khuller, Purohit, and Sarpatwar~\cite{khuller2020analyzing} investigated this problem when the profit function is a \emph{special submodular function} and designed a $\frac{12}{1-1/e}$-approximation algorithm. By generalizing the analysis of~\cite{khuller2020analyzing}, this factor was improved by Lamprou,  Sigalas, and Zissimopoulos~\cite{lamprou2020improved} to $\frac{11}{1-e^{-7/8}}$. They also showed that this problem cannot be approximated to within less than a factor of $\frac{e}{e-1}$, unless $P=NP$. 

Lee and Dooly~\cite{lee1996algorithms} provided a $(B-2)$-approximation algorithm for \UndirectedBudgetAdditive, where each vertex costs $1$. Johnson, Minkoff, and Phillips~\cite{johnson2000prize} introduced an edge-cost variant of \UndirectedBudgetAdditive, called \textbf{E\UndirectedBudgetAdditive}. They proposed a $(5+\epsilon)$-approximation algorithm for the unrooted version of \textbf{E\UndirectedBudgetAdditive} using Garg's $3$-approximation algorithm~\cite{garg3} for the $k$-MST problem, and observed that a $2$-approximation for $k$-MST would lead to a $3$-approximation for \textbf{E\UndirectedBudgetAdditive}. This observation by~\cite{johnson2000prize} along with the Garg's $2$-approximation algorithm~\cite{garg2005saving} for $k$-MST lead to a $3$-approximation algorithm for the unrooted version of \textbf{E\UndirectedBudgetAdditive}. A polynomial time $2$-approximation algorithm was proposed by~\cite{paul2020budgeted} for \textbf{E\UndirectedBudgetAdditive}. Using this $2$-approximation algorithm for \textbf{E\UndirectedBudgetAdditive} as a black-box,~\cite{d2022budgeted} showed that one can achieve a polynomial time $16\Delta/\epsilon^{2}$-approximation algorithm for \UndirectedBudgetAdditive that violates the budget constraint by a factor of $1+\epsilon$, where $\Delta$ is the maximum degree of the graph and $\epsilon \in (0, 1]$.

\cite{zhou2018relay} studied a variant of \textbf{E\UndirectedBudgetAdditive} in the wireless sensor networks and provided a $10$-approximation algorithm. \cite{seufert2010bonsai} investigated a special case of the unrooted version of \textbf{URAT}, where each vertex costs $1$ and we aim to find a tree with at most $B$ nodes maximizing the accumulated prize. Indeed, this is equivalent to the unrooted version of \textbf{E\UndirectedBudgetAdditive} when the cost of each edge is $1$ and we are looking for a tree containing at most $B-1$ edges to maximize the accumulated prize. \cite{seufert2010bonsai} provided a $(5+\epsilon)$-approximation algorithm for this problem. Similarly, \cite{huang2019maximizing} investigated this variant of \textbf{E\UndirectedBudgetAdditive} (or \UndirectedBudgetAdditive) in the plane and proposed a $2$-approximation algorithm.

The edge-cost quota variant of \UndirectedBudgetAdditive, called \textbf{EQ-URAT}, has been investigated by Johnson, Minkoff, and Phillips~\cite{johnson2000prize}. They showed that an $\alpha$-approximation algorithm for the $k$-MST problem results in an $\alpha$-approximation algorithm for \textbf{EQ-URAT}. Hence, one can have a $2$-approximation algorithm for \textbf{EQ-URAT} using the $2$-approximation algorithm of Garg~\cite{garg2005saving} for the $k$-MST problem. \cite{d2022budgeted} showed that using this $2$-approximation algorithm for \textbf{EQ-URAT} as a black-box, one can achieve a $2\Delta$-approximation algorithm for the quota variant \UndirectedBudgetAdditive (\textbf{Q-URAT}).

In the prize collecting variants of \UndirectedBudgetAdditive, called \textbf{PC-URAT}, the goal is to minimize the cost of the nodes in the resulting tree plus the prizes of vertices not spanned by the tree.
K{\"o}nemann and Sadeghian and Sanit{\`{a}} \cite{konemann2013lmp} designed a Lagrangian multiplier preserving $O(\ln{n})$-approximation algorithm for \textbf{PC-URAT}. Bateni, Hajiaghay and Liaghat~\cite{bateni2018improved} provided an $O(\log{n})$-approximation algorithm for a more general case of \textbf{PC-URAT}. There exists no $o(\ln{n})$-approximation algorithm for \textbf{PC-URAT}, unless $NP \subseteq \text{DTIME}(n^{\text{Polylog}(n)})$~\cite{klein1995nearly}. 
Goemans and Williamson~\cite{goemans1995general} provided a $2$-approximation algorithm for the edge cost variant of \textbf{PC-URAT}, called \textbf{EPC-URAT}. Later, \cite{archer2011improved} proposed a $(2-\epsilon)$-approximation algorithm for \textbf{EPC-URAT} which was a breakthrough upon the long-standing factor of $2$. \cite{HajiaghayiKKN12} considered a more general variant of \textbf{EPC-URAT} in which every prescribed pair $(u,v)$ should be connected through at least $r_{uv}$ edge-disjoint paths, otherwise we should pay the penalty. Furthermore, the penalty functions in their setting are monotone submodular functions. \cite{HajiaghayiKKN12} provided a constant-factor approximation algorithm for this general version of \textbf{EPC-URAT}.

\section{Notation and problem statement}

For an integer $s$, let $[s]:=\{1,\ldots,s\}$. Let $G=(V, E)$ (resp. $D=(V, A)$) be an undirected graph (resp. a directed graph) with a distinguished vertex $r \in V$ and $c:V \rightarrow \mathbb{R}^{\ge 0}$ be a nonnegative cost function on nodes. 

A \emph{path} is an undirected graph made of a sequence of distinct vertices $\{v_1, \dots, v_s\}$ and a sequence of edges $\{v_i, v_{i+1}\}$, $i\in[s-1]$. A \emph{tree} is an undirected graph in which any two vertices are connected by exactly one path.  If a subgraph $T$ of a graph $G$ is a tree, then we say that $T$ is a tree of $G$. 

A \emph{directed path} is a directed graph made of a sequence of distinct vertices $(v_1, \dots, v_s)$ and a sequence of directed edges $(v_i, v_{i+1})$, $i\in[s-1]$.
An \emph{out-tree} (a.k.a. out-arborescence) is a directed graph in which there is exactly one directed path from a specific vertex $r$, called \emph{root}, to each other vertex. If a subgraph $T$ of a directed graph $D$ is an out-tree, then we say that $T$ is an out-tree of $D$. 
For simplicity of reading, we will refer to out-trees simply as trees when it is clear that we are in the context of directed graphs.

Given two nodes $u$ and $v$ in $V$, the cost of a path from $u$ to $v$ in $G$ (resp. $D$) is the sum of the cost of its nodes. A path from $u$ to $v$ with the minimum cost is called a \textit{shortest path} and its cost, denoted by $dist(u,v)$, is called the \textit{distance} from $u$ to $v$ in $G$ (resp. $D$). Let $F$ denote the maximum distance from $r$ to a node in $V$, $F:=\max_{v\in V}\{ dist(r,v) \}$.

For any subgraph $G'$ of $G$ (resp. $D'$ of $D$), we denote by $V(G')$ and $E(G')$ (resp. $V(D')$ and $A(D')$) the set of nodes and edges in $G'$ (resp. $D'$), respectively. Consider a subset $S \subseteq V$. $G[S]$ (resp. $D[S]$) denotes the graph induced by set $S$, i.e., $E(G[S])=\{\{u, v\} \in E|u,v \in S\}$ (resp. $A(D[S])=\{(u, v) \in A|u,v \in S\}$).

Now we state the problems we consider on directed and undirected graphs.

\paragraph{Directed Graphs.}



Let $D=(V, A)$ be a directed graph with $n$ nodes, $c:V \rightarrow \mathbb{R}^{\ge 0}$ be a cost function on nodes,  $p:V \rightarrow \mathbb{R}^{\ge 0}$ be a prize function on nodes, $r\in V$ be a root vertex. We consider three variants of node-weighted Steiner problems.

\begin{enumerate}
    
    \item The Budgeted Directed Rooted Additive Tree problem (\DBudgetAdditiveName): Given a budget $B \in \mathbb{R}^{+}$, find an out-tree $T$ of $D$ rooted at $r$ that maximizes $p(T)=\sum_{v \in V(T)}p(v)$ subject to $c(T)=\sum_{v \in V(T)}c(v) \le B$.
    
    \item The Quota Directed Rooted Additive Tree problem (\DQuotaAdditiveName): Given a quota $Q \in \mathbb{R}^{+}$, find an out-tree $T$ of $D$ rooted at $r$ that minimizes $c(T)=\sum_{v \in V(T)}c(v)$ subject to $p(T)=\sum_{v \in V(T)}p(v) \ge Q$.
    
    \item The node-weighted Directed Steiner Tree problem (\DSteinerT): Given a set of terminals $K \subseteq V$, find an out-tree $T$ of $D$ rooted at $r$ such that $K\subseteq V(T)$ and $c(T)=\sum_{v\in V(T)}c(v)$ is minimum.
    
\end{enumerate}

\paragraph{Undirected Graphs.}

Let $G=(V, E)$ be an undirected graph with $n$ nodes, $c:V \rightarrow \mathbb{R}^{\ge 0}$ be a cost function on nodes, $p:2^V \rightarrow \mathbb{R}^{\ge 0}$ be a monotone non-decreasing submodular prize function on subsets of nodes, $r\in V$ be a root vertex. We consider two variants of node-weighted Steiner problems.

\begin{enumerate}
    
    \item The Budgeted Undirected Rooted Submodular problem (\UBudgetName): Given a budget $B \in \mathbb{R}^{+}$, find a tree $T$ of $G$ that maximizes $p(T)$ subject to $r \in V(T)$ and $c(T)=\sum_{v \in V(T)}c(v) \le B$.
    
    \item The Quota Undirected Rooted Submodular problem (\UQuotaName): Given a quota $Q \in \mathbb{R}^{+}$, find a tree $T$ of $G$ that minimizes $c(T)=\sum_{v \in V(T)}c(v)$ subject to $r \in V(T)$ and $p(T)\ge Q$.
\end{enumerate}

Consider a budgeted problem \textbf{BP} $\in \{\text{\DBudgetAdditiveName, \UBudgetName}\}$ (resp. a quota problem \textbf{QP} $\in \{\text{\DQuotaAdditiveName, \UQuotaName}\}$) with a budget $B$ (resp. quota $Q$). For $\alpha,\beta\geq 1$, a bicriteria $(\beta, \alpha)$-approximation algorithm for \textbf{BP} (resp. \textbf{QP}) is one that, for any instance $I_B$ (resp. $I_Q$) of the problem, returns a solution $Sol_{I_B}$ (resp. $Sol_{I_Q}$) such that $p(Sol_{I_B}) \ge \frac{OPT_{I_B}}{\alpha}$ (resp. $c(Sol_{I_Q}) \le \alpha OPT_{I_Q}$)  and $c(Sol_{I_B}) \le \beta B$ (resp. $p(Sol_{I_Q}) \ge \frac{Q}{\beta}$), where $OPT_{I_B}$ (resp. $OPT_{I_Q}$) is the optimum for $I_B$ (resp. $I_Q$).
\section{Results and Techniques}

In this section, we summarize our results and give an overview of the new techniques. 
We distinguish between problems on directed and undirected graphs.

\subsection{Directed Graphs}
We introduce a new technique to compute a tree with bounded cost and prize and we use it in all our algorithms for the problems on directed graphs. Since our algorithm for \DSteinerT presents our main technique in a more simpler way, in the following, we describe the core idea behind such technique by using the case of \DSteinerT as an illustrative example.

The next theorem, whose full proof is given in Section~\ref{sec:DST}, provides a polynomial time approximation algorithm for \DSteinerT.
\begin{restatable}{theorem}{thDSteinerTree}\label{thDSteinerTree}
For any $\epsilon >0$, \DSteinerT admits an $O\left((1+\epsilon)\sqrt{n} \ln{n}\right)$-approximation algorithm whose running time is polynomial in the input size and in $1/\epsilon$.
\end{restatable}

Let $I=<D=(V, A), c, r, K>$ be an instance of \DSteinerT. First note that we can suppose w.l.o.g. that any vertex $v \in V$ has a distance no more than $(1+\epsilon)c(T^*)$ from $r$ for any $\epsilon>0$, where $T^*$ is an optimal solution to~\DSteinerT.

 For every $v \in V$, we let $\mathcal{P}_v$ be the set of simple paths in $D$ from $r$ to $v$. For each $v\in V$, let $c_v=c(v)$.
We use the standard flow-based linear programming relaxation for \DSteinerT, which is as follows. 
\begin{align}\label{lpDSteinerTree}
 & \tag{LP-DST}\\
 \text{minimize}\quad \sum_{v \in V} x_v c_v& \\
  \text{subject to}\quad \sum_{P \in \mathcal{P}_t}f^t_P &= 1, &&\forall t \in K\label{lpDSteinerTree:overallflow}\\
                 \sum_{P \in \mathcal{P}_t:v \in P} f^t_{P}&\le x_v, &&\forall v \in V, t \in K\label{lpDSteinerTree:capacity}\\
                 0 \le &x_v \le 1, &&\forall v\in V\notag\\
                 0 \le &f^t_P \le 1, &&\forall t\in K, P \in \mathcal{P}_t\notag
\end{align}

For any terminal $t \in K$ and path $P \in \mathcal{P}_t$, $f^t_P$ denotes the amount of flow sent from $r$ to $t$ using $P$. Constraint~\eqref{lpDSteinerTree:overallflow} ensures that for any $t \in K$, the total amount of flow sent from $r$ to $t$ has to be equal to $1$. Constraint~\eqref{lpDSteinerTree:capacity} ensures that for any $t \in K$, the total amount flow sent from $r$ to $t$ passing through a vertex $v$ is at most $x_v$ for any $v \in V$. Note that, although the number of variables in~\eqref{lpDSteinerTree} is exponential in the input size, we can find an optimal solution in polynomial time as we need to find, independently for each $t$, the maximum flow from $r$ to $t$ by using $x_v$ as node capacities (see Appendix~\ref{apx:lp}).

It is easy to see that an optimal solution for~\eqref{lpDSteinerTree} provides a lower bound to $c(T^*)$. In fact, the solution to~\eqref{lpDSteinerTree} in which $x_v$ is set to 1 if $v\in V(T^*)$ and 0 otherwise, and $f^t_P$ is set to 1 if $P$ is the unique path from $r$ to $t$ in $T^*$ and to 0 otherwise, is feasible for~\eqref{lpDSteinerTree} and has value $\sum_{v \in V} x_v c_v=c(T^*)$.

Let $x$ be an optimal solution for~\eqref{lpDSteinerTree} and let $S\subseteq V$ be the set of all vertices $v$ with $x_v >0$ in~\eqref{lpDSteinerTree}. 

Let $U \subseteq S$ be the set of all nodes $v$ with $x_v \ge \frac{1}{\sqrt{n}}$ and $U'=S\setminus U$. We split the set $K$ of terminals into two sets: $CH$ is the set of terminals that are reachable from $r$ through at least a path $P$ that contains only vertices in $U$; and $EX$ is the set of terminals that are reachable from $r$ only through paths containing at least a node in $U'$, $EX=K\setminus CH$. 
We compute two trees $T^{CH}$ and $T^{EX}$ rooted at $r$ and spanning all terminals in $CH$ and $EX$, respectively, and bound their cost. Then we merge the two trees to obtain a directed Steiner tree of all terminals.

As $\sum_{v \in U} x_v c_v \le \sum_{v \in S} x_v c_v \le c(T^*)$ and $x_v \ge \frac{1}{\sqrt{n}}$ for any $v \in U$ (by definition), then $\sum_{v \in U} c_v \leq \sqrt{n}\cdot c(T^*)$ and hence we can span all terminals in $CH$ with a tree $T^{CH}$ rooted at $r$ that costs $c(T^{CH})\leq\sqrt{n}\cdot c(T^*)$.

Now we  show that we can compute a tree $T^{EX}$ rooted at $r$ spanning all the terminals in $EX$ with cost $c(T^{EX})=O\left((1+\epsilon)c(T^*)\sqrt{n}\ln{n}\right)$. The algorithm to build $T^{EX}$ can be summarized as follows. We first compute, for each $t\in EX$, the set $X_t$ of vertices $w$ in $U'$ for which there exists a path from $w$ to $t$ that uses only vertices in $U\cup \{w\}$. Then we compute a small-size hitting set $X'$ of all sets $X_t$. Finally, we connect $r$ to the vertices of $X'$ and the vertices of $X'$ to those in $EX$ in such a way that each node $t$ in $EX$ is reached from one of the vertices in $X'$ that hits $X_t$. The bound on the cost of $T^{EX}$ follows from the size of $X'$ and from the cost of nodes in $U$. 

Indeed, we can show that there exists a hitting set $X'$ of all sets $X_t$ with size $|X'| \le \sqrt{n}\ln{n}$. To show this, we first prove  (see Claim~\ref{clBound-X-t}) that $|X_t| \ge \sqrt{n}$ for any $t \in EX$ as follows. By constraint~\eqref{lpDSteinerTree:overallflow} of~\eqref{lpDSteinerTree},  each terminal should receive one unit of flow, moreover, by definition of $EX$ and $X_t$, any path from $r$ to any $t \in EX$ should pass through a vertex $w \in X_t$, which implies that any flow from $r$ to $t$ should also pass through a vertex $w \in X_t$. Therefore, one unit of flow needs to pass through vertices in $X_t$. Since, by definition of $U'$, each node $w \in X_t \subseteq U'$ has a capacity $x_w$ smaller than $\frac{1}{\sqrt{n}}$, we have $|X_t| \ge \sqrt{n}$. 



Since each $X_t$, $t\in EX$, has at least $\sqrt{n}$ elements and the size of $\bigcup_{t \in EX}X_t$ is at most $n$, then we can find in polynomial time a subset $X'$ of $\bigcup_{t \in EX}X_t$ that hits all the sets $X_t$ and has size at most $\sqrt{n}\ln{n}$ (see Claim~\ref{clHittingSet}).





For each $w \in X'$, we find a shortest path from $r$ to $w$. Since $|X'| \le \sqrt{n}\ln{n}$ and $dist(r, v) \le (1+\epsilon)c(T^*)$ for any $v \in V$, the total cost of these shortest paths from $r$ to $X'$ is at most $(1+\epsilon)\sqrt{n}\ln{n}\cdot c(T^*)$. We also find, for each $t\in EX$, a shortest path from an arbitrary vertex $w\in X' \cap X_t$ to $t$ through only vertices in $U$ (such a path exists by definition). The total cost of these shortest paths is at most $\sqrt{n}\cdot c(T^*)$ as they contain only vertices in $U$. This implies that there exists a tree $T^{EX}$ rooted at $r$ spanning the terminals in $EX$ with $c(T^{EX})=O((1+\epsilon)\sqrt{n}\ln{n}\cdot c(T^*))$.


Finally, as both $T^{EX}$ and $T^{CH}$ are rooted at $r$, we can find a tree $T$ rooted at $r$ that spans all vertices $V({T^{EX}})\cup V({T^{CH}})$ and has cost $c(T)=O\left((1+\epsilon) \sqrt{n}\ln{n}\cdot c(T^*)\right)$.


It is worth pointing out that by using the above algorithm for \DSteinerT as a black-box, one can obtain a polynomial time bicriteria $\left(2, O((1+\epsilon)n^{3/4}\ln{n})\right)$-approximation algorithm for \DQuotaAdditiveName, and a polynomial time bicriteria $\left(1+\epsilon, O(\frac{n^{3/4}\ln{n}}{\epsilon^2})\right)$-approximation algorithm for \DBudgetAdditiveName, where $\epsilon \in (0, 1]$.
We briefly explain how to obtain the approximation algorithm for \DQuotaAdditiveName, the one for \DBudgetAdditiveName is similar. Let $I_Q=<D=(V, A), c, p, r>$ be an instance of \DQuotaAdditiveName. First note that we can suppose w.l.o.g. that any vertex $v \in V$ has a distance no more than $(1+\epsilon)c(T^*_Q)$ from $r$, where $T^*_Q$ is an optimal solution to~\DQuotaAdditiveName. We first use a standard linear programming flow-based formulation to find an optimal fractional solution for \DQuotaAdditiveName (see~\eqref{lpDBudgetQuotaAdditive}). Let $x$ be an optimal solution for this linear program and $S=\{v \in V|x_v>0\}$. Next we partition $S$ into two subsets $S_1$ and $S_2$, where $S_1=\{v \in S|x_v \ge \frac{1}{n^{1/4}}\}$ and $S_2=S\setminus S_1$. Let $OPT$ be the fractional optimum for \DQuotaAdditiveName. 
If $\sum_{v \in S_1} x_v p(v) \ge Q/2$, then let $I=<D=(V, A), c, r, S_1>$ be an instance of \DSteinerT in which $S_1$ is the set of terminals. 
For the instance $I$, consider the solution $x'$ to~\eqref{lpDSteinerTree} in which $x'_v=1$ for any $v \in S_1$ and $x'_v=n^{1/4}\cdot x_v$ for any $v \in S_2$. This implies that an optimal fractional solution of~\eqref{lpDSteinerTree} obtained from $I$ costs at most $n^{1/4}\cdot OPT$ and all nodes in $S_1$ receive one unit of flow.
Now we can use our $O((1+\epsilon)\sqrt{n}\ln{n})$-approximation algorithm for \DSteinerT to find a tree spanning $S_1$. We have two subsets $CH, EX \subseteq S_1$ defined as above, so we can cover $CH$ and $EX$ by out-trees of cost at most $n^{3/4}OPT$ and $O(n^{3/4}\ln{n}(1+\epsilon)c(T^*_Q))$, respectively. 
If $\sum_{v \in S_2} x_v p(v) \ge Q/2$, then, by an averaging argument, we can find a subset of vertices with prize $Q/2$ and size at most $2n^{3/4}$. Hence we can span these $2n^{3/4}$ vertices by an out-tree of cost at most $2(1+\epsilon)n^{3/4}\cdot c(T^*_Q)$, for any $\epsilon >0$, as we can assume that no vertex has a distance more than $(1+\epsilon)c(T^*_Q)$ from $r$.

We can improve the approximation ratio for \DBudgetAdditiveName and \DQuotaAdditiveName by directly using our technique.
Our results for \DBudgetAdditiveName and \DQuotaAdditiveName are the following.

\begin{restatable}{theorem}{MainDBudgetTheorem}\label{thMainDBudget}
\DBudgetAdditiveName admits a polynomial time bicriteria $\left(1+\epsilon, O(\frac{n^{2/3}\ln{n}}{\epsilon^2})\right)$-approximation algorithm, for any $\epsilon \in (0, 1]$.
\end{restatable}


\begin{restatable}{theorem}{thDQuotaTree}\label{thDQuotaTree}
For any $\epsilon >0$, \DQuotaAdditiveName admits a bicriteria $\left(2, O((1+\epsilon)n^{2/3} \ln{n})\right)$-approximation algorithm whose running time is polynomial in the input size and in $1/\epsilon$. 
\end{restatable}

\subsection{Undirected Graphs}

The previous best polynomial time approximation algorithm for \UBudgetName was a bicriteria $\left(1+\epsilon, O(\frac{\sqrt{B}}{\epsilon^3})\right)$-approximation algorithm~\cite{d2022budgeted}. We improve this result when $B \ge n^{1+\alpha}$ for any $\alpha > 0$ and provide a new  algorithm for \UQuotaName.

\begin{restatable}{theorem}{MainUBudgetTheorem}\label{thMainUBudget}
\UBudgetName admits a polynomial time bicriteria $\left(1+\epsilon, O(\frac{\sqrt{n}\log{n}}{\epsilon^3})\right)$-approximation algorithm, for any $\epsilon \in (0, 1]$.
\end{restatable}

\begin{restatable}{theorem}{thUQuotaTree}\label{thUQuotaTree}
Let $T^*_Q$ be an optimal solution to \UQuotaName. \UQuotaName admits an algorithm that, for any $\epsilon >0$, outputs a tree $T$ for which one of the following two conditions holds: (i) $c(T) = O(c(T^*_Q)\sqrt{n}\log{n})$ and $p(T)\geq \frac{Q}{2}$;  (ii) $c(T) \le (1+\epsilon)c(T^*_Q)$ and $p(T)\geq \frac{Q}{2\sqrt{n}}$; and whose running time is polynomial in the input size and in $1/\epsilon$.
\end{restatable}

To show Theorems~\ref{thMainUBudget} and~\ref{thUQuotaTree}, we use the same core idea. In what follows, we explain such an idea by describing our algorithm for \UBudgetName.

Our approach extends to monotone submodular prize functions the technique which was introduced by Guha, Moss, Naor, and Schieber~\cite{guha1999efficient} for the case of the budgeted problem on undirected graphs and additive prize function.
Their algorithm consists of three steps: (i) it computes an optimal solution to a linear relaxation of their problem, which provides an upper bound $OPT$ on the optimal prize; (ii) by using this optimal fractional solution, it computes a tree $T$ which may violate the budget constraints but has a prize of at least $OPT/4$ and a prize-to-cost ratio $\gamma = \Omega\left(\frac{OPT}{B\log^2{n}}\right)$; (iii) it applies a trimming process that takes as input $T$ and returns another tree which violates the budget constraint by a factor of at most $2$ but preserves a prize-to-cost ratioof $\gamma$ (up to a bounded multiplicative factor) and costs at least $B/2$. Therefore, the obtained tree guarantees a bicriteria $(2, O(\log^2{n}))$-approximation.

We encounter two main issues if we want to use the technique provided by~\cite{guha1999efficient} in general monotone submodular functions:

\begin{enumerate}
    \item The linear program provided by~\cite{guha1999efficient} may not give an upper bound to the optimum prize of \UBudgetName. To overcome this issue, we provide a new linear programming formulation whose optimum is an upper bound on the optimum prize of \UBudgetName and which can be solved in polynomial time through a separation oracle.
    
    \item The trimming procedure by~\cite{guha1999efficient} cannot be applied to the case of general monotone submodular functions. To deal with this issue, we introduce a new trimming process by refining the one provided by D'Angelo, Delfaraz and Gilbert~\cite{d2022budgeted} for the case of monotone submodular prizes. 
    The prize-to-cost ratio of the tree produced by this latter is smaller than that of the original tree by an arbitrarily large multiplicative factor.
    Here we show how to reduce this factor to a fixed constant, while keeping upper and lower bounds on the cost. This results in a tree that guarantees an approximation factor of $O(\frac{\sqrt{n}\log{n}}{\epsilon^3})$ at the cost of a budget violation of $1+\epsilon$, for any $\epsilon\in (0,1]$.
\end{enumerate}


\section{Node-Weighted Directed Steiner Tree}\label{sec:DST}


As a warm-up,
we present a polynomial time approximation algorithm for \DSteinerT using the standard flow-based linear programming relaxation~\eqref{lpDSteinerTree}.

\thDSteinerTree*

We prove Theorem~\ref{thDSteinerTree} in what follows.
We denote an optimal solution to \DSteinerT by $T^*$. Recall that $dist(v, u)$ denotes the cost of a shortest path from $v$ to $u$ in $D$.

We first argue that we can assume that for all nodes $v$ 
we have $dist(r, v) \leq (1+\epsilon)c(T^*)$, for any $\epsilon>0$. 
We remove from the graphs all the nodes that do not satisfy this condition by estimating the value of $c(T^*)$ using the following procedure.

Let $c_{\min}$ be the minimum positive cost of a vertex and $c_{M}$ be the cost of a minimum spanning tree of $D$, we know that $c(T^*)\leq c_{M}$. We estimate the value of $c(T^*)$ by guessing $N$ possible values, where $N$ is the smallest integer for which $c_{\min}(1+\epsilon)^{N-1}\geq c_{M}$.

For each guess $i\in [N]$, we remove the nodes $v$ with $dist(r, v)> c_{\min}(1+\epsilon)^{i-1}$
, and compute a Steiner Tree in the resulting graph, if it exists, with an algorithm that will be explained later. Eventually, we output the computed Steiner Tree with the smallest cost. Since $c_{\min}(1+\epsilon)^{N-2}< c_{M}$, the number $N$ of guesses is smaller than $\log_{1+\epsilon}(c_{M}/c_{\min}) + 2$, which is polynomial in the input size and in $1/\epsilon$.

Let $i\in [N]$ be the smallest value for which $c_{\min}(1+\epsilon)^{i-1}\geq c(T^*)$. Then, $c(T^*)> c_{\min}(1+\epsilon)^{i-2}$ and for all the nodes $v$ in the graph used in guess $i$, we have $dist(r, v)\leq c_{\min}(1+\epsilon)^{i-1} <(1+\epsilon)c(T^*)$. 
Since we output the solution with the minimum cost among those computed in the guesses for which our algorithm returns a feasible Steiner Tree, then the final solution will not be worse than the one computed at guess $i$. Therefore, from now on we focus on guess $i$ and assume that $dist(r, v) \leq (1+\epsilon)c(T^*)$
, for all nodes $v$.

We now describe our $O\left((1+\epsilon)\sqrt{n} \ln{n}\right)$-approximation algorithm under this assumption. 

Let $x$ be an optimal solution for~\eqref{lpDSteinerTree} and let $S\subseteq V$ be the set of all vertices $v$ with $x_v >0$ in~\eqref{lpDSteinerTree}. Note that $\sum_{v \in S}x_v c_v \le c(T^*)$ as~\eqref{lpDSteinerTree} provides a lower bound on the optimum solution of \DSteinerT. Let $U \subseteq S$ be the set of all nodes with $x_v \ge \frac{1}{\sqrt{n}}$ for any $v \in U$. Note that $r$ belongs to $U$ since we need to send one unit of flow from $r$ to any terminal by constraint~\eqref{lpDSteinerTree:overallflow}.
We call a terminal $t \in K$ a \emph{cheap terminal} if there exists a path from $r$ to $t$ in $D[U]$. We call a terminal $t \in K$ an \emph{expensive terminal} otherwise. Let $CH$ and $EX$ be the set of all cheap and expensive terminals in $K$, respectively.

We show that there exists a tree $T^{CH}$ rooted at $r$ spanning all the cheap terminals $CH$ with cost $c(T^{CH})\leq\sqrt{n}\cdot c(T^*)$.

\begin{lemma}\label{clSpanning-Cheap-Terminals}
There exists a polynomial time algorithm that finds a tree $T^{CH}$ rooted at $r$ spanning all the cheap terminals $CH$ with cost $c(T^{CH})\leq\sqrt{n}\cdot c(T^*)$.
\end{lemma}
\begin{proof}
By definition, each terminal $t$ in $CH$ is reachable from $r$ through some paths $P$ that contains only vertices in $U$, i.e., $V(P) \subseteq U$. Thus we compute a shortest path $P$ from $r$ to every $t \in CH$ in $D[U]$. Let $\mathcal{P}^{CH}$ be the union of all these shortest paths $P$ and $V(\mathcal{P}^{CH})$ be the set of all vertices of the paths in $\mathcal{P}^{CH}$. Now we find a tree $T^{CH}$ rooted at $r$ spanning all the vertices $V(\mathcal{P}^{CH})$. By construction, $T^{CH}$ contains all terminals in $CH$.

We know that

\[
\sum_{v \in V(T^{CH})} x_v c_v = \sum_{v \in V(\mathcal{P}^{CH})} x_v c_v \le \sum_{v \in U} x_v c_v \le \sum_{v \in S} x_v c_v \le c(T^*),
\]
where the equality is due to $V(T^{CH}) = V(\mathcal{P}^{CH})$, the first inequality is due to $V(\mathcal{P}^{CH}) \subseteq U$ and the second inequality is due to $U \subseteq S$.

By definition of $U$, we have $x_v \ge \frac{1}{\sqrt{n}}$ for any $v \in U$, then 
\[
c(T^{CH})=\sum_{v \in V(T^{CH})}c_v\leq\sqrt{n}\cdot\sum_{v \in V(T^{CH})} x_v c_v  \leq \sqrt{n}\cdot c(T^*),
\]
which concludes the proof. 
\end{proof}

We next show that there exists a tree $T^{EX}$ rooted at $r$ spanning all the expensive terminals $EX$ with cost $c(T^{EX})=O\left((1+\epsilon)c(T^*)\sqrt{n}\ln{n}\right)$. The algorithm to build $T^{EX}$ can be summarized as follows. We first compute, for each $t\in EX$, the set $X_t$ of vertices $w$ in $S\setminus U$ for which there exists a path $P$ from $w$ to $t$ that uses only vertices in $U\cup \{w\}$, i.e., $V(P)\setminus \{w\} \subseteq U$. Then, we compute a small-size hitting set $X'$ of all $X_t$. Finally, we connect $r$ to the vertices of $X'$ and the vertices of $X'$ to those in $EX$ in such a way that each node $t$ in $EX$ is reached from one of the vertices in $X'$ that hits $X_t$. The bound on the cost of $T^{EX}$ follows from the size of $X'$ and from the cost of nodes in $U$.

\begin{lemma}\label{clSpanning-Expensive-Terminals}
There exists a polynomial time algorithm that finds a tree $T^{EX}$ rooted at $r$ spanning all the cheap terminals $EX$ with cost $c(T^{EX})\leq 2 (1+\epsilon)\sqrt{n}\ln{n}\cdot c(T^*)$.
\end{lemma}
\begin{proof}
Let $U' \subseteq S$ be the set of all vertices $v$ with $0<x_v <\frac{1}{\sqrt{n}}$, i.e., $U'=S\setminus U$. 

We now show that $|X_t| \ge \sqrt{n}$ for any $t \in EX$. 

\begin{claim}\label{clBound-X-t}
$|X_t| \ge \sqrt{n}$, for each $t \in EX$.
\end{claim}
\begin{proof}

We know that (i) each terminal must receive one unit of flow (by constraint~\eqref{lpDSteinerTree:overallflow} of~\eqref{lpDSteinerTree}), (ii) any path $P$ from $r$ to any $t \in EX$ in the graph $D[S]$ contains at least one vertex $w \in U'$ (by definition of expensive terminals), and , (iii) in any path $P$ from $r$ to any $t \in EX$, the node $w \in U'$ in $P$ that is closest to $t$ is a member of $X_t$, i.e. $w \in X_t$ (by definition of $X_t$), therefore any flow from $r$ to $t$ must pass through a vertex $w \in X_v$. This implies that the vertices in $X_t$ must send one unit of flow to $v$ in total. Since each of them can only send at most $1/\sqrt{n}$ amount of flow, they must be at least $\sqrt{n}$. Formally, we have 
\[
 1 = \sum_{P\in\mathcal{P}_t} f^t_P \leq \sum_{w \in X_t}\sum_{P\in\mathcal{P}_t:w\in P} f^t_P\leq \sum_{w \in X_t} x_w \leq \sum_{w \in X_t} \frac{1}{\sqrt{n}} =\frac{|X_t|}{\sqrt{n}},
\]
which implies that $|X_t| \ge \sqrt{n}$. Note that the first equality follows from constraint~\eqref{lpDSteinerTree:overallflow} of~\eqref{lpDSteinerTree}, the first inequality is due to the fact that, by definition of $X_t$, any path $P$ from $r$ to a $t \in EX$ contains a vertex $w \in X_t$, the second inequality is due to constraint~\eqref{lpDSteinerTree:capacity} of~\eqref{lpDSteinerTree}, the last inequality 
is due to $x_w<\frac{1}{\sqrt{n}}$, for each $w\in U'$, and $X_t\subseteq U'$.
This concludes the proof of the claim.
\end{proof}
We use the following well-known result (see e.g. Lemma~3.3 in~\cite{Chan2007} and Appendix~\ref{apx:claim} for a proof) to find a small set of vertices that hits all the sets $X_t$, for all $t\in EX$.
\begin{claim}\label{clHittingSet}
Let $V'$ be a set of $M$ elements and $\Sigma=(X'_1, \dots, X'_N)$ be a collection of subsets of $V'$ such that $|X'_i|\ge R$, for each $i\in [N]$.
There is a deterministic algorithm which runs in polynomial time in $N$ and $M$ and finds a subset $X' \subseteq V'$ with $|X'|\le (M/R)\ln{N}$ and $X' \cap X'_i\ne \emptyset$ for all $i\in[N]$.
%
\end{claim}

Thanks to Claim~\ref{clBound-X-t}, we can use the algorithm of Claim~\ref{clHittingSet} to find a set $X'\subseteq \bigcup_{t \in EX}X_t$ such that $X' \cap X_t\ne \emptyset$, for all $t\in EX$, whose size is at most 
\[
|X'| \le \frac{n\ln{n}}{\sqrt{n}}\leq \sqrt{n}\ln{n},
\]
where the parameters of Claim~\ref{clHittingSet} are $R=\sqrt{n}$, $N=|EX|\leq n$, and $M=\big|\bigcup_{t \in EX}X_t\big| \leq |V| \leq n$.

In other words, since for any $t \in EX$ and any $w \in X_t$ there exists a path from $w$ to $t$ in $D[U \cup \{w\}]$ and $X' \cap X_t\ne \emptyset$, then there exists at least a vertex $w\in X'$ for which there is a path from $w$ to $t$ in $D[U \cup \{w\}]$.

Now, for each $w \in X'$, we find a shortest path from $r$ to $w$ in $D$. Let $\mathcal{P}_1$ be the set of all these shortest paths. 
We also select, for each $t\in EX$, an arbitrary vertex $w$ in $X'\cap X_t$ and compute a shortest path from $w$ to $t$ in $D[U\cup \{w\}]$. Let $\mathcal{P}_2$ be the set of all these shortest paths. 
Let $V(\mathcal{P}_1)$ and $V(\mathcal{P}_2)$ denote the union of all vertices of the paths in $\mathcal{P}_1$ and $\mathcal{P}_2$, respectively, and let $D^{EX}$ be the graph induced by all the vertices in $V(\mathcal{P}_1)\cup V(\mathcal{P}_2)$.

Now we find a tree $T^{EX}$ rooted at $r$ spanning $D^{EX}$. Note that such a tree exists as in $D^{EX}$ we have for each $w \in X'$ a path from $r$ to $w$ and, for each terminal $t\in EX$, at least a path from one of the vertices in $X'$ to $t$.

We next move to bounding the cost of $T^{EX}$, indeed we bound the cost of all vertices in $D^{EX}$. 
Since $|X'|\le \sqrt{n}\ln{n}$ and the cost of a shortest path from $r$ to any $v \in V$ in $D$ is at most $(1+\epsilon)c(T^*)$, then $c(V(\mathcal{P}_1))\leq (1+\epsilon)\sqrt{n}\ln{n}\cdot c(T^*)$. 
Since $x_v \ge \frac{1}{\sqrt{n}}$ for any $v \in U$, and $\sum_{v \in U}x_v c_v\leq \sum_{v \in S}x_v c_v \leq c(T^*)$, then $c(U)=\sum_{v \in U}c_v\leq \sqrt{n} \cdot c(T^*)$. Therefore, since $V(\mathcal{P}_2)\setminus X'\subseteq U$, then $c(V(\mathcal{P}_2)\setminus X')\leq c(U)\leq \sqrt{n} \cdot c(T^*)$.
Overall, $D^{EX}$ costs at most $2(1+\epsilon)\sqrt{n}\ln{n}\cdot c(T^*)$. This finishes the proof.
\end{proof}

Now we prove Theorem~\ref{thDSteinerTree}.

\begin{proof}[Proof of Theorem~\ref{thDSteinerTree}]
Since both $T^{EX}$ and $T^{CH}$ are rooted at $r$, we can find a tree $T$ rooted at $r$ spanning all vertices $V({T^{EX}})\cup V({T^{CH}})$. By Lemmas~\ref{clSpanning-Cheap-Terminals} and~\ref{clSpanning-Expensive-Terminals} we have $c(T)=O\left((1+\epsilon) c(T^*)\sqrt{n}\ln{n}\right)$. This concludes the proof.
\end{proof}

\section{Directed Rooted Additive Tree Problems}

In this section we present a polynomial time bicriteria $\left(2, O((1+\epsilon)n^{2/3} \ln{n})\right)$-approximation algorithm for \DQuotaAdditiveName and a polynomial time bicriteria $\left(1+\epsilon, O(\frac{n^{2/3}\ln{n}}{\epsilon^2})\right)$-approximation algorithm for \DBudgetAdditiveName, where $\epsilon$ is an arbitrary number in $(0,1]$.
Through the section we let $I_Q=<D=(V, A), c, p, r, Q>$ and $I_B=<D=(V, A), c, p, r, B>$ be two instances of \DQuotaAdditiveName and \DBudgetAdditiveName, respectively, and we let $T^*_Q$ and $T^*_B$ be two optimal solutions for $I_Q$ and $I_B$, respectively. 
Both algorithms use the same technique and can be summarized in the following three steps: 

\begin{enumerate}
    \item We define a set of linear constraints, denoted as~\eqref{lpDBudgetQuotaAdditive}, over fractional variables, that takes a given quota $Q$ and a given budget $B$ as parameters and admits a feasible solution if there exists a subtree $T$ of $D$ rooted at $r$ such that $c(T)\leq B$ and $p(T)\geq Q$.
    Observe that the minimum $B$ for which~\eqref{lpDBudgetQuotaAdditive} is feasible for a given $Q$ is a lower bound on $c(T^*_Q)$, while the maximum $Q$ for which~\eqref{lpDBudgetQuotaAdditive} is feasible for a given $B$ is an upper bound to $p(T^*_B)$.
    \item We give a polynomial time algorithm that takes as input a feasible solution to~\eqref{lpDBudgetQuotaAdditive} and computes a subtree $T$ of $D$ rooted at $r$ such that $c(T) = O((F+B)n^{2/3}\ln{n})$ and $p(T)\geq \frac{Q}{2}$, where $F$ is the maximum distance from $r$ to any other node, $F:=\max_{v\in V}\{ dist(r,v) \}$.

    \item 
    \begin{description}
    \item[\DQuotaAdditiveName:] We first show that we can assume that $F\leq (1+\epsilon)c(T^*_Q)$, for any $\epsilon>0$. Then, we use a solution for~\eqref{lpDBudgetQuotaAdditive} that minimizes $B$ as input to the algorithm in the previous step and obtain a tree $T$ such that $c(T) = O((F+B)n^{2/3}\ln{n}) = O((1+\epsilon)c(T^*_Q)n^{2/3}\ln{n})$ and $p(T)\geq \frac{Q}{2}$.
    \item[\DBudgetAdditiveName:] We assume w.l.o.g. that $F\leq B$ and use a solution for~\eqref{lpDBudgetQuotaAdditive} that maximizes $Q$ as input to the algorithm in the previous step and obtain a tree $T$ such that $c(T) = O(Bn^{2/3}\ln{n})$ and $p(T)\geq \frac{Q}{2}\geq \frac{p(T^*_B)}{2}$. 
    
    Tree $T$ may violate the budget constraint by a large factor, however the ratio $\gamma$ between its prize and its cost is $\Omega\left(\frac{p(T^*_B)}{Bn^{2/3}\ln{n}}\right)$. Therefore, we can apply to $T$ the trimming process given in~\cite{bateni2018improved} to obtain another tree $\hat T$ with cost $\frac{\epsilon}{2}B\le c(\hat T)\le (1+\epsilon)B$, for any $\epsilon \in (0, 1]$, and prize-to-cost ratio $\frac{p(\hat T)}{c(\hat T)}=\frac{\epsilon \gamma}{4}$. The obtained tree $\hat{T}$ achieves an approximation ratio of $O(\frac{n^{2/3}\ln{n}}{\epsilon^2})$ at the cost of a budget violation of $1+\epsilon$, for any $\epsilon \in (0, 1]$.
    \end{description}

\end{enumerate}


In the following, we will detail each step of our algorithms.

\paragraph{Bounding the optimal cost and prize.}

Here we define a set of linear constraints that admits a feasible solution if there exists a tree $T$ rooted in $r$ in $D$ such that $c(T)\leq B$ and $p(T)\geq Q$, for given parameters $B$ and $Q$.

For each $v \in V$, let $p_v=p(v)$, $c_v=c(v)$, and $\mathcal{P}_v$ be the set of simple paths in $D$ from $r$ to $v$. Our set of constraints~\eqref{lpDBudgetQuotaAdditive} is defined as follows.

\begin{align}\label{lpDBudgetQuotaAdditive}
& \tag{Const-DRAT}\\
 \sum_{v \in V} x_v p_v &\ge Q\label{lpDBudgetQuotaAdditive:quota}\\
 \sum_{v \in V} x_v c_v &\le B\label{lpDBudgetQuotaAdditive:budget}\\
                  \sum_{P \in \mathcal{P}_v}f^v_P &= x_v,&& \forall v \in V\setminus \{r\}\label{lpDBudgetQuotaAdditive:overrallflow}\\
                 \sum_{P \in \mathcal{P}_v:w \in P} f^v_{P}&\le x_w,&& \forall v\in V\setminus \{r\} \text{ and }\forall w \in V\setminus \{ v\}\label{lpDBudgetQuotaAdditive:capacity}\\
                 0 \le x_v &\le 1, &&\forall v\in V\notag\\
                 0 \le f^v_P &\le 1, &&\forall v\in V\setminus \{r\}, P \in \mathcal{P}_v\notag
\end{align}

We use variables $f^v_P$ and $x_v$, for each $v\in V$ and $P \in \mathcal{P}_v$, where $f^v_P$ represents the amount of flow sent from $r$ to $v$ using path $P$ and $x_v$ represents both the capacity of node $v$ and the overall amount of flow sent from $r$ to $v$.
Variables $x_v$, for $v\in V$, are called \emph{capacity variables}, while variables $f_v^P$ for $v\in V$ and $P \in \mathcal{P}_v$ are called \emph{flow variables}. 


The constraints in~\eqref{lpDBudgetQuotaAdditive} are as follows. Constraints~\eqref{lpDBudgetQuotaAdditive:quota} and~\eqref{lpDBudgetQuotaAdditive:budget}  ensure that any feasible (fractional) solution to~\eqref{lpDBudgetQuotaAdditive} has a prize at least $Q$ and a cost at most $B$. Constraints~\eqref{lpDBudgetQuotaAdditive:overrallflow} and~\eqref{lpDBudgetQuotaAdditive:capacity}  formulate a connectivity constraint through standard flow encoding, that is they ensure that the nodes $v$ with $x_v>0$ induce subgraph in which all nodes are reachable from $r$. In particular, constraint~\eqref{lpDBudgetQuotaAdditive:overrallflow} ensures that the amount of flow that is sent from $r$ to any vertex $v$ must be equal to $x_v$ and constraint~\eqref{lpDBudgetQuotaAdditive:capacity} ensures that the total flow from $r$ to $v$ passing through a vertex $w$ cannot exceed $x_w$.

Note that~\eqref{lpDBudgetQuotaAdditive} has an exponential number of variables. However, it can be solved efficiently as we only need to find, independently for any $v \in V\setminus \{r\}$, a flow from $r$ to $v$ of value $x_v$ that does not exceed the capacity $x_w$, for each vertex $w \in V\setminus\{r,v\}$ (see Appendix~\ref{apx:lp}).

We now show that a feasible solution to~\eqref{lpDBudgetQuotaAdditive} can be used to find a lower bound to an optimal solution for \DQuotaAdditiveName and an upper bound to an optimal solution for \DBudgetAdditiveName.
In particular, we show that, for any tree $T$ rooted in $r$ in $D$ such that $c(T)\leq B$ and $p(T)\geq Q$ we can compute a feasible solution $x$ for~\eqref{lpDBudgetQuotaAdditive}.

\begin{lemma}\label{lmAdditiveLPOtimpality}
Given a directed graph $D=(V, A)$, $r\in V$,  $c:V \rightarrow \mathbb{R}^{\ge 0}$, $p:V \rightarrow \mathbb{R}^{\ge 0}$, $B\in \mathbb{R}^{\ge 0}$ and $Q\in \mathbb{R}^{\ge 0}$, if there exists a tree $T$ rooted in $r$ in $D$ such that $c(T)\leq B$ and $p(T)\geq Q$, then there exists a feasible solution $x$ for~\eqref{lpDBudgetQuotaAdditive}.
\end{lemma}
\begin{proof}
Let us consider a solution to~\eqref{lpDBudgetQuotaAdditive} in which $x_v=1$ for all $v\in V(T)$, while $x_v$ is set to $0$ for all $v\not\in V(T)$.
As $p(T)\ge Q$ and $c(T)\le B$, then the quota and budget constraints~\eqref{lpDBudgetQuotaAdditive:quota}--\eqref{lpDBudgetQuotaAdditive:budget} are satisfied. Since $T$ is connected and for any $v \in V(T)$, there exists only one path $P$ from $r$ to $v$ in $T$, constraints~\eqref{lpDBudgetQuotaAdditive:overrallflow} and~\eqref{lpDBudgetQuotaAdditive:capacity} are satisfied by setting $f^v_P=1$ and any other flow variable to 0.
\end{proof}

Lemma~\ref{lmAdditiveLPOtimpality} shows the existence of a feasible solution $x$ to~\eqref{lpDBudgetQuotaAdditive} when the budget is $B=c(T^*_Q)$ and the quota is $Q$, therefore the optimum $OPT_Q$ to the linear program of minimizing $B$ subject to constraints~\eqref{lpDBudgetQuotaAdditive}
gives a lower bound to $c(T^*_Q)$, i.e. $OPT_Q \leq c(T^*_Q)$.
Similarly, the optimum $OPT_B$ to the linear program of maximizing $Q$ subject to constraints~\eqref{lpDBudgetQuotaAdditive}
gives an upper bound to $p(T^*_B)$, $OPT_B \geq p(T^*_B)$. We denote these two linear programs by \DQuotaAdditiveLP and \DBudgetAdditiveLP, respectively.

\paragraph{Finding a good tree from a feasible fractional solution to~\eqref{lpDBudgetQuotaAdditive}.}
Here we elaborate the second step and show how to find a tree with a good trade-off between prize and cost by using a feasible fractional solution from the previous step. In particular, we will show the following theorem. Recall that $F$ denotes the maximum distance from $r$ to a node in $V$, $F:=\max_{v\in V}\{ dist(r,v) \}$.
\begin{theorem}\label{th:tree-ratio}
Given a feasible solution $x$ to~\eqref{lpDBudgetQuotaAdditive}, then there exists a polynomial time algorithm that computes a tree $T$ rooted at $r$ such that  $c(T) = O((F+B)n^{2/3}\ln{n})$ and $p(T)\geq \frac{Q}{2}$.
\end{theorem}

For \DBudgetAdditiveName, we can assume w.l.o.g. that $F\leq B$, therefore the above theorem implies that the ratio between prize and cost of the computed tree is $\Omega\left( \frac{Q}{Bn^{2/3}\ln{n}} \right)$. Similarly, for the \DQuotaAdditiveName problem we can assume that $F\leq(1+\epsilon)B$ and the above theorem implies that the prize-to-cost ratio of the computed tree is $\Omega\left( \frac{Q}{(1+\epsilon)Bn^{2/3}\ln{n}} \right)$.

We now prove Theorem~\ref{th:tree-ratio}.
Let $x$ be a feasible solution for~\eqref{lpDBudgetQuotaAdditive} and let $S \subseteq V$ be the set of vertices $v$ with $x_v>0$, i.e., $S=\{v \in V: x_v>0\}$.
We partition $S$ into two subsets $S_1, S_2 \subseteq S$, where $S_1=\{v \in S| x_v\ge \frac{1}{n^{1/3}}\}$ and $S_2=\{v \in S| x_v < \frac{1}{n^{1/3}}\}$.

We first focus on nodes in $S_1$ and, in the following lemma, we show how to compute a tree $T$ rooted at $r$ spanning all vertices in $S_1$ with cost $c(T)=O((F+B)n^{2/3}\ln{n})$.

\begin{lemma}\label{thSpanning-Sz}
There exists a polynomial time algorithm that finds a tree $T$ rooted at $r$ spanning all vertices in $S_1$ with cost $c(T)=O((F+B)n^{2/3}\ln{n})$.
\end{lemma}
\begin{proof}
Let $U:=\{v \in S : x_v \ge\frac{1}{ n^{2/3}}\}$. We call a vertex $v \in S_1$ a \emph{cheap vertex} if there exists a path from $r$ to $v$ in $D[U]$ (the graph induced by $U$). We call a vertex $v \in S_1$ an \emph{expensive vertex} otherwise. Note that $r$ belongs to $U$ when $U\ne \emptyset$ since we need to send $\frac{1}{n^{2/3}}$ amount of flow from $r$ to any vertex in $U$ by constraint~\eqref{lpDBudgetQuotaAdditive:overrallflow}. Let $CH$ and $EX$ be the set of all cheap and expensive vertices in $S_1$, respectively.

In the following, we first show that we can compute in polynomial time  two trees $T^{CH}$ and $T^{EX}$ spanning all nodes in $CH$ and $EX$, respectively, and having cost $c(T^{CH})=O(Bn^{2/3})$ and $c(T^{EX})=O((F+B) n^{2/3} \ln{n})$, then we show how to merge the two trees into a single tree with cost $O((F+B) n^{2/3} \ln{n})$. 

We first focus on $T^{CH}$. By definition, all vertices in $CH$ are reachable from $r$ through paths that contain only vertices in $U$. Thus, for each $v \in CH$, we compute a shortest path from $r$ to $v$ in $D[U]$ and find a tree $T^{CH}$ rooted at $r$ spanning all and only the vertices in the union of these shortest paths. Since $\sum_{u \in U} x_u c_u \le B$ (by constraint~\eqref{lpDBudgetQuotaAdditive:budget}) and $x_u \ge \frac{1}{n^{2/3}}$ for any $u \in U$ (by definition of $U$), then $c(U)= \sum_{u \in U} c_u\leq Bn^{2/3}$. Hence $c(T^{CH}) \leq c(U) =O(Bn^{2/3})$.

Now we show that there exists a tree $T^{EX}$ rooted at $r$ spanning all the expensive vertices $EX$ with cost $c(T^{EX})=O((F+B) n^{2/3} \ln{n})$ and that we can compute $T^{EX}$ in polynomial time.
The algorithm to build $T^{EX}$ can be summarized as follows. We first compute, for each $v\in EX$, the set $X_v$ of vertices $w \in S\setminus U$ for which there exists a path from $w$ to $v$ that uses only vertices in $U\cup \{w\}$. Then we compute a small-size hitting set $X'$ of all $X_v$. Finally, we connect $r$ to the vertices of $X'$ and the vertices of $X'$ to those in $EX$ in such a way that each node $v$ in $EX$ is reached from one of the vertices in $X'$ that hits $X_v$. The bound on the cost of $T^{EX}$ follows from the size of $X'$ and from the cost of nodes in $U$. We now detail on the construction of $T^{EX}$ and its cost analysis.

Let $U' \subseteq S$ be the set of all vertices $w$ with $x_w <\frac{1}{n^{2/3}}$, i.e., $U'=\{w \in S:x_w <\frac{1}{n^{2/3}}\}$ and $U'=S\setminus U$.
For any expensive vertex $v \in EX$, we define $X_v$ as the set of vertices $w$ in $U'$ such that there exists a path from $w$ to $v$ in $D[U \cup \{w\}]$.
Note that for any $w \in X_v$, $v$ is reachable from $w$ through a path $P$ such that $V(P)\setminus \{w\} \subseteq U$, i.e., $V(P)\setminus \{w\}$ only contains vertices from $U$.

The following claim  gives a lower bound on the size of $X_v$, for each $v \in EX$. We follow similar arguments as those of Claim~\ref{clBound-X-t}, but we include it for the sake of completeness. 
\begin{claim}\label{clBoundXv}
$|X_v| \ge n^{1/3}$, for each $v \in EX$.
\end{claim}
\begin{proof}
We know that (i) the amount of flow that each vertex $v \in S_1$ should receive is at least $\frac{1}{n^{1/3}}$ (by definition of $S_1$ and constraint~\eqref{lpDBudgetQuotaAdditive:overrallflow} of~\eqref{lpDBudgetQuotaAdditive}), (ii) any path $P$ from $r$ to any $v \in EX$ in the graph $D[S]$ contains at least one vertex $w \in U'$ (by definition of expensive vertices), and , (iii) in any path $P$ from $r$ to any $v \in EX$,  the node $w \in U'$ in  $P$ that is closest to $v$ is a member of $X_v$, i.e. $w \in X_v$ (by definition of $X_v$), therefore any flow from $r$ to $v$ should pass through a vertex $w \in X_v$. This implies that the vertices in $X_v$ must send at least $\frac{1}{n^{1/3}}$ amount of flow to $v$ in total . Formally, we have 

\[
\frac{1}{n^{1/3}} \le x_v= \sum_{P\in\mathcal{P}_v} f^v_P \leq \sum_{w \in X_v}\sum_{P\in\mathcal{P}_v:w\in P} f^v_P\leq \sum_{w \in X_v} x_w \leq \sum_{w \in X_v} \frac{1}{n^{2/3}} =\frac{|X_v|}{n^{2/3}},
\]
which implies that $|X_v| \ge n^{1/3}$. Note that the first inequality follows from the definition of $S_1$, the first equality follows from constraint~\eqref{lpDBudgetQuotaAdditive:overrallflow} of~\eqref{lpDBudgetQuotaAdditive}, the second inequality is due to the fact that, by definition of $X_v$, any
path $P\in\mathcal{P}_v$ contains a vertex $w \in X_v$, the third inequality is due to constraint~\eqref{lpDBudgetQuotaAdditive:capacity}  of~\eqref{lpDBudgetQuotaAdditive}, the last inequality 
is due to $x_w < \frac{1}{n^{2/3}}$, for each $w \in U'$, and $X_v\subseteq U'$.
This concludes the proof of the claim.
\end{proof}

Using the the bound of Claim~\ref{clBoundXv} on the size of sets $X_v$, we can exploit the algorithm of Claim~\ref{clHittingSet} to find a set $X'\subseteq \bigcup_{v \in EX}X_v$ such that $X' \cap X_v\ne \emptyset$, for all $v\in EX$, whose size is at most
\[
|X'| \le \frac{n \ln{n}}{n^{1/3}}=n^{2/3}\ln{n},
\]
where in this case the parameters of Claim~\ref{clHittingSet} are $R=n^{1/3}$, $M=\big|\bigcup_{v \in EX}X_v\big|\leq n$, and $N=|EX|\leq n$.
In other words, since for any $v \in EX$ there exists a path from any $w \in X_v$ to $v$ in $D[U \cup \{w\}]$ and $X'$ contains at least one vertex in $X_v$, then there exists a path from one of the vertices $w \in X'$ to $v$ in $D[U \cup \{w\}]$.

Now we find a shortest path from $r$ to any $w \in X'$ in $D$. Let $\mathcal{P}_1$ be the set of all these shortest paths. We also find, for each $v\in EX$, a shortest path from an arbitrary vertex $w\in X' \cap X_v$ to $v$ in $D[U\cup \{w\}]$ (the choice of $w$ can be made arbitrarily if there are several vertices $w$ in $X' \cap X_v$). Let $\mathcal{P}_2$ be the set of all these shortest paths. Let $V(\mathcal{P}_1)$ and $V(\mathcal{P}_2)$ be the union of all the nodes of the paths in $\mathcal{P}_1$ and $\mathcal{P}_2$, respectively. Now we find a tree $T^{EX}$ rooted at $r$ spanning graph $D^{EX}:=D[V(\mathcal{P}_1) \cup V(\mathcal{P}_2)]$. Note that such a tree exists as in $D^{EX}$ there exists a path from $r$ to any $w \in X'$ and, for each vertex $v\in EX$, at least a path from one of the vertices in $X'$ to $v$.

We next move to bounding the cost of $T^{EX}$, i.e. we bound the total cost of nodes in $V(\mathcal{P}_1) \cup V(\mathcal{P}_2)$. 
Since $|X'|\le n^{2/3}\ln{n}$ and the maximum distance from $r$ to any other node is $F$,  then $c(V(\mathcal{P}_1))\leq Fn^{2/3}\ln{n}$. 
As $\sum_{u \in U}c(v)\leq Bn^{2/3}$ (by constraint~\eqref{lpDBudgetQuotaAdditive:budget} of~\eqref{lpDBudgetQuotaAdditive} and $x_u \ge \frac{1}{n^{2/3}}$ for any $u \in U$) and $V(\mathcal{P}_2)\setminus X'\subseteq U$, then $c(V(\mathcal{P}_2)\setminus X')\leq Bn^{2/3}$. 
Overall, $T^{EX}$ costs at most $O((F+B)n^{2/3}\ln{n})$. 

We have shown that $T^{CH}$ and $T^{EX}$ span all nodes in $CH$ and $EX$, respectively, and cost $c(T^{CH})=O((F+B)n^{2/3})$ and $c(T^{EX})=O((F+B)n^{2/3}\ln{n})$. Since both $T^{EX}$ and $T^{CH}$ are rooted at $r$, we can find a tree $T$ rooted at $r$ that spans all vertices $V(T^{EX})\cup V(T^{CH})$. Since $c(T^{CH}) + c(T^{EX}) = O((F+B)n^{2/3}\ln{n})$, we have $c(T)=O((F+B)n^{2/3}\ln{n})$, which concludes the proof.
\end{proof}

We are now ready to prove Theorem~\ref{th:tree-ratio}. 
\begin{proof}[Proof of Theorem~\ref{th:tree-ratio}]
Since $\sum_{v \in S} x_v p_v\geq Q$, we know that $\sum_{v \in S_1}x_v p_v \ge \frac{Q}{2}$ or  $\sum_{v \in S_2}x_v p_v \ge \frac{Q}{2}$.

If $\sum_{v \in S_1}x_v p_v \ge \frac{Q}{2}$, then by Lemma~\ref{thSpanning-Sz}, we can find in polynomial time a tree $T$ that spans all vertices in $S_1$ such that $c(T)=O((F+B)n^{2/3}\ln{n})$. Since $T$ spans all vertices of $S_1$, then $p(T)\ge p(S_1)\ge \frac{Q}{2}$.

If $\sum_{v \in S_2}x_v p_v \ge \frac{Q}{2}$, we have that 
\begin{align}\label{eq-Case2LowerBoundPrize-Additive}
    p(S_2)= \sum_{v \in S_2}p_v\ge n^{1/3}\cdot\sum_{v \in S_2}x_v p_v \ge \frac{n^{1/3} Q}{2},
\end{align}
where the first inequality holds since $0<x_v <\frac{1}{n^{1/3}}$ for any $v \in S_2$ and the second inequality holds by the case assumption.
We partition $S_2$ into $M$ groups $U_1, \dots, U_M$ in such a way that for each $i\in [M-1]$, $|U_i|= 2|S_2|^{2/3}$, and  $|U_M|\le 2|S_2|^{2/3}$. Hence the number of selected groups is at most
\[
 M \le \left\lceil \frac{|S_2|}{2|S_2|^{2/3}} \right\rceil = \left\lceil \frac{|S_2|^{1/3}}{2} \right\rceil \le \left\lfloor \frac{|S_2|^{1/3}}{2}\right\rfloor + 1 \le |S_2|^{1/3} \le  n^{1/3}.
\]

Now among $U_1, \dots, U_M$, we select the group $U_{z}$ that maximizes the prize, i.e., $z=\arg\max_{i \in [M]}p(U_i)$. We know that
\[
p(U_{z}) \ge\frac{1}{M}\sum_{i=1}^{M}p(U_i)= \frac{p(S_2)}{M}\ge\frac{n^{1/3}Q}{2n^{1/3}} =\frac{Q}{2},
\]
where the first inequality is due to averaging argument, the first equality is due to the additivity of $p$, and the second inequality is due to $M\le n^{1/3}$ and Inequality~\eqref{eq-Case2LowerBoundPrize-Additive}.
   
We now find for each vertex $v$ in $U_z$ a shortest path from $r$ to $v$ and compute a tree $T$ that spans all the vertices in the union of these shortest paths. Clearly, $c(T) \le 2F n^{2/3}$ as $|U_z| \le 2n^{2/3}$ and the cost of a shortest path from $r$ to any $v \in V$ in $G$ is at most $F$. Furthermore, $p(T) \ge p(U_z) \ge \frac{Q}{2}$, by additivity of $p$. This concludes the proof.
\end{proof}

We next show how to use Theorem~\ref{th:tree-ratio} to devise bicriteria approximation algorithms for \DQuotaAdditiveName and \DBudgetAdditiveName.

\paragraph{Approximation algorithm for \DQuotaAdditiveName.}
Recall that $dist(v, u)$ denotes the cost of a shortest path from $v$ to $u$ in $D$ and that $T^*_Q$ is an optimal solution to $I_Q$.

Like in the case of \DSteinerT, we can assume that for all nodes $v$ we have $dist(r, v) \leq (1+\epsilon)c(T^*_Q)$, i.e. that $F\leq (1+\epsilon)c(T^*_Q)$, for any $\epsilon >0$. 
We hence describe our bicriteria $O\left(2, (1+\epsilon)n^{2/3} \ln{n}\right)$-approximation algorithm for \DQuotaAdditiveName under this assumption.

We first find an optimal solution $x$ to~\DQuotaAdditiveLP, let $OPT_Q$ be the optimal value of~\DQuotaAdditiveLP. Observe that $x$ is a feasible solution for the set of constraints~\eqref{lpDBudgetQuotaAdditive} in which $B=OPT_Q\leq c(T^*_Q)$ (see Lemma~\ref{lmAdditiveLPOtimpality}). Therefore, we can use $x$ and apply the algorithm in Theorem~\ref{th:tree-ratio} to obtain a tree $T$ such that $c(T) = O((1+\epsilon)c(T^*_Q)n^{2/3}\ln{n})$ and $p(T)\geq \frac{Q}{2}$. This shows the following theorem.
\thDQuotaTree*

\paragraph{Approximation algorithm for \DBudgetAdditiveName.}
Let us assume that for every $v\in V$, $dist(r, v)\le B$, i.e. $F\leq B$, since otherwise we can remove from $D$ all the nodes $v$ such that $dist(r, v)> B$.

Let $x$ be an optimal solution for \DBudgetAdditiveLP and let $T$ be a tree computed from $x$ by the algorithm in Theorem~\ref{th:tree-ratio}. Since $x$ is a feasible solution to~\eqref{lpDBudgetQuotaAdditive} when $Q=OPT_B \geq p(T^*_B)$, then the prize of $T$ is at least $\frac{p(T^*_B)}{2}$ but its cost can exceed the budget $B$. In this case, however, the cost of $T$ is bounded by $c(T)=O(Bn^{2/3} \ln{n})$ and its prize-to-cost ratio is $\gamma=\frac{p(T)}{c(T)}=\Omega\left(\frac{p(T^*_B)}{B n^{2/3}\ln{n}}\right)$. Therefore, we can use $T$ and a variant of the trimming process introduced by Bateni, Hajiaghay and Liaghat~\cite{bateni2018improved} for undirected graphs, to compute another tree $\hat T$ with cost between $\frac{\epsilon B}{2}$ and $(1+\epsilon)B$ and prize-to-cost ratio $\frac{\epsilon \gamma}{4}$, for any $\epsilon\in (0,1]$. The resulting tree violates the budget at most by a factor $1+\epsilon$ and guarantees an approximation ratio of $O(\frac{n^{2/3}\ln{n}}{\epsilon^2})$.

The trimming process given by Bateni, Hajiaghay and Liaghat~\cite{bateni2018improved} has been used for the undirected version of \DBudgetAdditiveName.
In particular, in their case, we are given an undirected graph $G=(V, E)$, a distinguished vertex $r \in V$ and a budget $B$, where each vertex $v \in V$ is assigned with a prize $p'(v)$ and a cost $c'(v)$. For a tree $T$, the prize and cost of $T$ are the sum of the prizes and costs of the nodes of $T$ and are denoted by $p'(T)$ and $c'(T)$, respectively. A graph $G$ is called \emph{$B$-proper} for the vertex $r$ if the cost of reaching any vertex from $r$ is at most $B$. \cite{bateni2018improved} proposed a trimming process that leads to the following lemma.

\begin{lemma}[Lemma 3 in \cite{bateni2018improved}]\label{lmBateniTrimmingProcess}
Let $T$ be a tree rooted at $r$ with the prize-to-cost ratio $\gamma=\frac{p'(T)}{c'(T)}$. Suppose the underlying graph is $B$-proper for $r$ and for $\epsilon \in (0, 1]$ the cost of the tree is at least $\frac{\epsilon B}{2}$. One can find a tree $T'$ containing $r$ with the prize-to-cost ratio at least $\frac{\epsilon \gamma}{4}$ such that $\epsilon B/2 \le c'(T') \le (1+\epsilon)B$.
\end{lemma}

We can use Lemma~\ref{lmBateniTrimmingProcess} in directed graphs and achieve the same guarantee. Let us first briefly explain the trimming process by~\cite{bateni2018improved}. Their trimming process takes as input a subtree $T$ rooted at a node $r$ of a $B$-proper graph $G$ and first (i) computes a subtree $T''$ of $T$, not necessarily rooted $r$, such that $\frac{\epsilon B}{2}\leq c'(T'') \leq \eps B$ and $p'(T'')\geq \frac{\eps B}{2}\gamma$, where $\gamma = \frac{p'(T)}{c'(T)}$. Then (ii) it connects node $r$ to the root of $T''$ with a minimum-cost path and obtains a tree $T'$ rooted at $r$. Since $G$ is $B$-proper, this path exists and has length at most $B$, which implies that the cost of the resulting tree is between $\frac{\epsilon B}{2}$ and $(1+\eps)B$, and its prize to cost ratio is at least $\frac{p'(T'')}{(1+\eps)B}\geq\frac{\epsilon \gamma}{4}$. This leads to Lemma~\ref{lmBateniTrimmingProcess}. For the case of directed graphs, we prove a lemma equivalent to Lemma~\ref{lmBateniTrimmingProcess}.

\begin{lemma}\label{coTrimmingProcess}
Let $D= (V, A)$ be a $B$-proper graph for a node $r$. Let $T$ be an out-tree of $D$ rooted at $r$ with the prize-to-cost ratio $\gamma=\frac{p(T)}{c(T)}$. Suppose that for $\epsilon \in (0, 1]$, $c(T) \ge \frac{\epsilon B}{2}$. One can find an out-tree $\hat{T}$ rooted at $r$ with the prize-to-cost ratio at least $\frac{\epsilon \gamma}{4}$ such that $\epsilon B/2 \le c'(\hat{T}) \le (1+\epsilon)B$.
\end{lemma}
\begin{proof}
We first define an undirected tree by ignoring the directions of edges in $T$. Now we apply to $T$ the step (i) of the trimming procedure by~\cite{bateni2018improved} with $p'(\cdot) = p(\cdot)$ and $c'(\cdot) =c(\cdot)$ as prize and cost functions, respectively. This enables us to find a subtree $T''$ such that $\frac{\epsilon B}{2}\leq c'(T'') \leq \eps B$ and $p'(T'')\geq \frac{\eps B}{2}\gamma$. We orient all the edges in $T''$ toward the leaves, which are the same directions as in $T$, to obtain an out-tree $T''_{out}$ corresponding to $T''$. Then, we add a minimum cost path from $r$ to the root of $T''_{out}$. Note that as $D$ is $B$-proper for $r$, such a path exists and has length at most $B$. The obtained out-tree $\hat{T}$ has the desired properties because $p(T''_{out}) = p'(T'') $ and $ c(T''_{out}) = c'(T'')$ and therefore $\frac{\epsilon B}{2}\leq c(\hat{T})\leq (1+\eps)B$ and $\frac{p(\hat{T})}{c(\hat{T})}\geq\frac{p'(T'')}{(1+\eps)B}\geq\frac{\epsilon \gamma}{4}$.
\end{proof}

This results in the following theorem.

\MainDBudgetTheorem*
\begin{proof}
We first find an optimal solution $x$ to~\DBudgetAdditiveLP, then, we use $x$ as input to the algorithm in Theorem~\ref{th:tree-ratio} to obtain a tree $T$ in which $c(T) = O(Bn^{2/3}\ln{n})$ and $p(T)\geq \frac{OPT_B}{2}\geq\frac{p(T^*_B)}{2}$ as discussed above. The prize-to-cost ratio of $T$ is $\gamma=\frac{p(T)}{c(T)}=\Omega\left(\frac{p(T^*_B)}{B n^{2/3}\ln{n}}\right)$. Then, if $c(T)>B$, we can apply the algorithm of Lemma~\ref{coTrimmingProcess} to $T$ and compute another tree $\hat T$ with cost $c(\hat T)\le (1+\epsilon)B$ and prize-to-cost ratio at least $\frac{\epsilon \gamma}{4}$. Moreover, $c(\hat T) \ge \epsilon B/2$, and therefore we have $p(\hat T)=\Omega\left( \frac{\epsilon^2  p(T^*_B)}{n^{2/3}\ln{n}}\right)$, which concludes the proof.
\end{proof}

\section{Undirected Rooted Submodular Tree problems}
In this section we present our polynomial time bicriteria approximation algorithms for \UQuotaName and \UBudgetName.
Let $I_Q=<G=(V, E), c, p, r, Q>$ and $I_B=<G=(V, E), c, p, r, B>$ be two instances of \UQuotaName and \UBudgetName, respectively, and let $T^*_Q$ and $T^*_B$ be two optimal solutions for $I_Q$ and $I_B$, respectively. 

The algorithms use a technique similar to that used for \DQuotaAdditiveName and \DBudgetAdditiveName, which can be summarized in the following three steps: 

\begin{enumerate}
    \item We define a set of linear constraints, denoted as~\eqref{lpUndirected-SubmodularFlow}, over fractional variables, that takes a given quota $Q$ and a given budget $B$ as parameters and admits a feasible solution if there exists a subtree $T$ of $G$ such that $r\in V(T)$, $c(T)\leq B$, and $p(T)\geq Q$.  The set of constraints uses an exponential number of linear constraints and variables but we can show that it can be solved in polynomial time through a separation oracle.


    \item We give a polynomial time algorithm that takes as input a feasible solution to~\eqref{lpUndirected-SubmodularFlow} and computes a subtree $T$ of $G$ such that $r\in V(T)$ and one of the following two conditions holds: (i) $c(T) = O(B\sqrt{n}\ln{n})$ and $p(T)\geq \frac{Q}{2}$ or (ii) $c(T)\leq F$ and $p(T)\geq \frac{Q}{2\sqrt{n}}$.

    \item 
    \begin{description}
    \item[\UQuotaName:] We first show that we can assume that $F\leq (1+\epsilon)c(T^*_Q)$, for any $\epsilon>0$. Then, we use a solution for~\eqref{lpUndirected-SubmodularFlow} that minimizes $B$ as input to the algorithm in the previous step and obtain a tree $T$ such that (i) $c(T) = O(c(T^*_Q)\sqrt{n}\log{n})$ and $p(T)\geq \frac{Q}{2}$;  or (ii) $c(T) \le (1+\epsilon)c(T^*_Q)$ and $p(T)\geq \frac{Q}{2\sqrt{n}}$.
    
    \item[\UBudgetName:] We assume w.l.o.g. that $F\leq B$ and use a solution for~\eqref{lpUndirected-SubmodularFlow} that maximizes $Q$ as input to the algorithm in the previous step and obtain a tree $T$ such that $c(T) = O(B\sqrt{n}\ln{n})$ and $p(T)\geq \frac{p(T^*_B)}{2}$ or (ii) $c(T)\leq B$ and $p(T)\geq \frac{p(T^*_B)}{2\sqrt{n}}$.
    
    In the first case, the computed tree $T$ can violate the budget constraint by a factor $h=O(\sqrt{n}\ln{n})$ but it has a prize-to-cost ratio $\gamma=\Omega\left(\frac{p(T^*_B)}{B\sqrt{n}\ln{n}}\right)$. We introduce a variant of the trimming process proposed by D'Angelo, Delfaraz and Gilbert~\cite{d2022budgeted} for submodular prize functions that allows us to compute another tree $\hat T$ such that $r\in V(T)$ and one of the two following conditions holds: the prize-to-cost ratio of $\hat{T}$ is at least $\frac{\epsilon^2 \gamma}{640}$ and $\epsilon B/2 \le c(\hat{T}) \le (1+\epsilon)B$; $p(\hat{T}) \ge p(T)/5h$ and $c(\hat{T}) \le B$.
    Therefore, the obtained tree $\hat{T}$ achieves an approximation ratio of $O(\frac{\sqrt{n}\ln{n}}{\epsilon^3})$ at the cost of a budget violation of $1+\epsilon$, for any $\epsilon \in (0, 1]$.
    
    The new trimming process might be of its own interest.
    \end{description}

\end{enumerate}
In the following, we will detail each step of our algorithms.

\paragraph{Bounding the optimal cost and submodular prize.}
We formulate a set of linear constraints, encoding a submodular flow problem, that admits a feasible solution if there exists a tree $T$ such that $r\in V(T)$, $c(T)\leq B$, and $p(T)\geq Q$.

We use the submodular flow problem introduced by Edmonds and Giles~\cite{edmonds1977min}, which is a generalization of the network flow problem in which restrictions on flows into vertices are generalized to flows into subsets of vertices. In particular, let $D=(V, A)$ be a directed graph with upper and lower capacity bounds $\bar c, \underbar c \in \mathbb{R}^A$ on its edges. Let $d:A \rightarrow \mathbb{R}$ be a weight function on the edges and $f:2^V \rightarrow \mathbb{R}$ be a submodular function on the subsets of vertices. The \emph{base polyhedron} of $f$ is defined as:

\[
\mathcal{B}(f)=\{x| x\in \mathbb{R}^V, x(X)\le f(X), \forall X \subset V, x(V)=f(V)\},
\]
where for each $X \subseteq V$ and $x\in \mathbb{R}^V$, $x(X)=\sum_{v \in X} x(v)$.

For any $S\subseteq V$, let $\delta^+(S)$ (resp. $\delta^-(S)$) be the set of edges entering (resp. leaving) the set $S$. Let $\upvarphi: A \rightarrow \mathbb{R}$ be a flow function. For any $X \subseteq V$, we define $\partial\upvarphi(X)=\sum_{e \in \delta^-(X)} \upvarphi(e) - \sum_{e \in \delta^+(X)} \upvarphi(e)$.
Consider the following linear program:

\[
\max \sum_{e \in A} d(e)\upvarphi(e) \text{ subject to } \underbar c (e)\le \upvarphi(e) \le \bar c(e) \forall e \in A, \partial\upvarphi \in  \mathcal{B}(f).
\]
A feasible solution $\upvarphi$ that satisfies both constraints in this linear program is called a submodular flow~\cite{fujishige2000algorithms}. Gr{\"o}tschel, Lov{\'a}sz, and Schrijver~\cite{grotschel1981ellipsoid} provided a separation oracle and used the Ellipsoid method to solve it in polynomial time. 

Now we are ready to provide our set of linear constraints, which is denoted as~\eqref{lpUndirected-SubmodularFlow}.
Let $D=(V, A)$ be the directed graph obtained from $G=(V, E)$ such that $A=\{(w,v), (v,w)|\{w, v\} \in E\}$.
Let $p_v=p(\{v\})$ and $c_v=c(v), \forall v \in V$. For every $v \in V$, we let $\mathcal{P}_v$ be the set of simple paths in $D$ from $r$ to $v$.
\begin{align}\label{lpUndirected-SubmodularFlow}
& \tag{Const-URST}\\
                  \sum_{v \in S} x_v p_v & \le p(S),&& \forall S\subseteq V\label{lpUndirected-SubmodularFlow:submodular}\\
                  \sum_{v \in V} x_v p_v &\geq Q\label{lpUndirected-SubmodularFlow:quota}\\
                  \sum_{v \in V} x_v c_v &\le B\label{lpUndirected-SubmodularFlow:budget}\\
                  \sum_{P \in \mathcal{P}_v}f^v_P &= x_v,&& \forall v \in V\setminus \{r\}\label{lpUndirected-SubmodularFlow:overallflow}\\
                 \sum_{P \in \mathcal{P}_v:w \in P} f^v_{P}&\le n x_w,&& \forall v\in V\setminus \{r\} \text{ and }\forall w \in V\setminus \{ v\}\label{lpUndirected-SubmodularFlow:capacity}\\
                 0 \le x_v  &\le 1, &&\forall v\in V\notag\\
                 0 \le f^v_P &\le 1, &&\forall v\in V\setminus \{r\}, P \in \mathcal{P}_v\notag
\end{align}

We use variables $f^v_P$ and $x_v$, for each $v\in V$ and $P \in \mathcal{P}_v$, where $f^v_P$ is the amount of flow sent from $r$ to $v$ using path $P$ and $x_v$ is the overall amount of flow sent from $r$ to $v$. 
The value of $x_v$ can also be intended as the fraction of submodular prize $p(\{v\})$ shipping from $v$ to a fictitious sink on a fictitious edge between $v$ and the sink.
Variables $x_v$, for $v\in V$, are called \emph{capacity variables}, while variables $f_v^P$ for $v\in V$ and $P \in \mathcal{P}_v$ are called \emph{flow variables}. 

The constraints in~\eqref{lpUndirected-SubmodularFlow} are as follows. Constraint~\eqref{lpUndirected-SubmodularFlow:submodular} enables us to move from submodular prizes to additive prizes. Constraint~\eqref{lpUndirected-SubmodularFlow:quota} and~\eqref{lpUndirected-SubmodularFlow:budget} ensure that any feasible solution has a prize at least $Q$ and a cost at most $B$, here the prize is intended as a linear combination of the singleton prizes where the coefficients are given by the capacity variables. 
Constraints~\eqref{lpUndirected-SubmodularFlow:overallflow} and~\eqref{lpUndirected-SubmodularFlow:capacity} formulate a connectivity constraint through standard flow encoding, that is they ensure that the nodes $v$ with $x_v>0$ induce subgraph in which all nodes are reachable from $r$. In particular, constraint~\eqref{lpUndirected-SubmodularFlow:overallflow} ensures that the overall amount of flow that is sent from $r$ to any vertex $v$ is equal to $x_v$ and constraint~\eqref{lpUndirected-SubmodularFlow:capacity} ensures that the total flow from $r$ to $v$ passing through a vertex $w$ cannot exceed $n x_w$. It will be clear later why in a flow from $r$ to $v$ the capacity of each vertex $w \in V \setminus \{r, v\}$ is set to $n x_w$.

Note that although the number of flow variables in~\eqref{lpUndirected-SubmodularFlow} is exponential, we can have an equivalent formulation with a polynomial number of variables (see Appendix~\ref{apx:lp}). Moreover, even if the number of constraints~\eqref{lpUndirected-SubmodularFlow:submodular} is exponential, we can use the separation oracle by Gr{\"o}tschel, Lov{\'a}sz, and Schrijver~\cite{grotschel1981ellipsoid} to solve~\eqref{lpUndirected-SubmodularFlow} in polynomial time.

The next theorem allows us to use a feasible solution to~\eqref{lpUndirected-SubmodularFlow} in order to compute a lower bound to an optimal solution for \UQuotaName and an upper bound to an optimal solution for \UBudgetName.

\begin{theorem}\label{th-UndirectedLPOtimpality}
Given a graph $G=(V, E)$, $r\in V$,  $c:V \rightarrow \mathbb{R}^{\ge 0}$, $p:2^V \rightarrow \mathbb{R}^{\ge 0}$, $B\in \mathbb{R}^{\ge 0}$ and $Q\in \mathbb{R}^{\ge 0}$, if there exists a tree $T$ in $G$ such that $r\in V(T)$, $c(T)\leq B$ and $p(T)\geq Q$, then there exists a feasible solution $x$ for~\eqref{lpUndirected-SubmodularFlow}.
\end{theorem}
\begin{proof}
Let us consider a variant~\eqref{restrictedlpUndirected-SubmodularFlow} of~\eqref{lpUndirected-SubmodularFlow} in which we only consider the graph induced by $V(T)$ and $x_v$ is constrained to be at least $1/n$, for all $v\in V(T)$.

\begin{align}\label{restrictedlpUndirected-SubmodularFlow}
& \tag{Const-URST'}\\
                   \sum_{v \in S} x_v p_v & \le p(S),&& \forall S\subseteq V(T)\label{constraint:submodular}\\
                   \sum_{v \in V(T)} x_v p_v &\geq Q\label{constraint:quota}\\
                  \sum_{v \in V(T)} x_v c_v &\le B\label{constraint:budget}\\
                  \sum_{P \in \mathcal{P}_v}f^v_P &= x_v,&&\forall v \in V\setminus \{r\}\label{constraint:flow}\\
                 \sum_{P \in \mathcal{P}_v:w \in P} f^v_{P}&\le n x_w ,&& \forall v\in V(T)\setminus \{r\} \text{ and }\forall w \in V(T)\setminus \{v\}\label{constraint:flowconservation}\\
                 1/n \le x_v &\le 1, &&\forall v\in V(T)\label{constraint:capacitypositive}\\
                 0 \le f^v_P&\le 1, &&\forall v\in V(T), P \in \mathcal{P}_v\label{constraint:flowpositive}
\end{align}
We observe that any feasible solution to~\eqref{restrictedlpUndirected-SubmodularFlow} induces a feasible solution for~\eqref{lpUndirected-SubmodularFlow} as it is enough to set all the additional variables in~\eqref{lpUndirected-SubmodularFlow} to be equal to 0 (recall that function $p$ is monotone non-decreasing).

We therefore construct a feasible solution to~\eqref{restrictedlpUndirected-SubmodularFlow}.
We first focus on constraints~\eqref{constraint:submodular} and show that we can construct a solution $x$ that satisfies all such constraints and for which $\sum_{v \in V(T)}x_v p_v = p(T)$, which implies that also constraint~\eqref{constraint:quota} is satisfied as $p(T)\geq Q$ by hypothesis.
We then show that the constructed solution satisfies also constraints~\eqref{constraint:budget}--\eqref{constraint:flowpositive}.

We start from solution $x$, where $x_v=1/n$, for each $v\in V(T)$. Since $p_v\leq p(S)$, for any $v\in S$ and $S\subseteq V(T)$, then constraint~\eqref{constraint:submodular} for set $S$ is satisfied for any $x_v\leq1/|S|$. By assigning $x_v = 1/n$, for all $v\in V(T)$, we have $x_v\leq 1/|S|$, for all $S\subseteq V(T)$, and all constraints~\eqref{constraint:submodular} are satisfied. We now show how to increase some of the variables $x_v$ in such a way that $\sum_{v \in V(T)}x_v p_v = p(T)$, i.e. that constraints~\eqref{constraint:submodular} corresponding to set of $V(T)$ is tight\footnote{Given a set $S\subseteq V(T)$, we say that constraints~\eqref{constraint:submodular} is tight for $S$ if $\sum_{v \in S} x_v p_v  = p(S)$.} and all the other constraints~\eqref{constraint:submodular} remain satisfied.

Starting from $x_v=1/n$, for each $v\in V(T)$, we increase variables $x_v$, one by one, until either $x_v=1$ or, for some set $S$ with $v\in S$, we have $\sum_{v \in S} x_v p_v = p(S)$. We keep on increasing variables until, for all $v\in V(T)$, we have $x_v = 1$ or $\sum_{v \in S} x_v p_v = p(S)$, for some $S$ such that $v\in S$. Note that, by using this procedure we have that constraint~\eqref{constraint:submodular} is tight for at least one set since, by submodularity of $p$, it holds $\sum_{v \in S} p_v \ge p(S)$, for all $S\subseteq V(T)$. Moreover, all constraints~\eqref{constraint:submodular} remain satisfied by construction.

The following lemma shows that if constraints~\eqref{constraint:submodular} corresponding to two subsets of $V(T)$ are tight, then the constraints corresponding to their union and their intersection are also tight.
\begin{lemma}\label{lmClosedUnsderIntersectionAndUnion}
Let $x$ be a solution satisfying constraints~\eqref{constraint:submodular} and let $U_1$ and $U_2$ be two subsets of $V(T)$ such that $\sum_{v \in U_i}x_v p_v =p(U_i)$ for any $i \in \{1,2\}$, then $\sum_{v \in U_1 \cup U_2}x_v p_v =p(U_1 \cup U_2)$ and $\sum_{v \in U_1 \cap U_2}x_v p_v =p(U_1 \cap U_2)$.
\end{lemma}
\begin{proof}
We have
\[
0=p(U_1)-\sum_{v \in U_1}x_v p_v+ p(U_2)-\sum_{v \in U_2}x_v p_v\ge p(U_1 \cup U_2)-\sum_{v \in U_1 \cup U_2}x_v p_v+ p(U_1 \cap U_2)-\sum_{v \in U_1 \cap U_2}x_v p_v.
\]
Note that the inequality holds because $p(U_1)+ p(U_2) \ge p(U_1 \cup U_2)+ p(U_1 \cap U_2)$ by the submodularity of $p$ and, clearly, $\sum_{v \in U_1}x_v p_v+ \sum_{v \in U_2}x_v p_v= \sum_{v \in U_1 \cup U_2}x_v p_v+ \sum_{v \in U_1 \cap U_2}x_v p_v$.

 By constraint~\eqref{constraint:submodular}, $p(U_1 \cup U_2)-\sum_{v \in U_1 \cup U_2}x_v p_v \ge 0$ and $p(U_1 \cap U_2)-\sum_{v \in U_1 \cap U_2}x_v p_v \ge 0$. This implies that $p(U_1 \cup U_2)=\sum_{v \in U_1 \cup U_2}x_v p_v$ and $p(U_1 \cap U_2)=\sum_{v \in U_1 \cap U_2}x_v p_v$, which concludes the proof.
\end{proof}
Let $S_1$ be the union of all sets for which the corresponding constraint~\eqref{constraint:submodular} is tight. By Lemma~\ref{lmClosedUnsderIntersectionAndUnion} we have that $\sum_{v \in S_1}x_v p_v =p(S_1)$. We now show that $S_1=V(T)$.
By contradiction, let us suppose that $S_2 = V(T)\setminus S_1$ is not empty. We have that
$\sum_{v \in S_2}x_v p_v < p(S_2)$, as otherwise we should have $S_2\subseteq S_1$, by definition of $S_1$. Therefore, there exists a set $S_3$ such that $S_2\cap S_3\neq \emptyset$ and $\sum_{v \in S_3}x_v p_v =p(S_3)$, as otherwise we can increase one of the variables corresponding to nodes in $S_2$, which is not possible by the construction of $x$. By definition of $S_1$, it follows that $S_3\subseteq S_1$, a contradiction to $S_1\cap S_2 = \emptyset$ and $S_2\cap S_3\neq \emptyset$.
This concludes the proof that $x$ satisfies constraints~\eqref{constraint:submodular} and~\eqref{constraint:quota}.

We now show that $x$ satisfies also constraints~\eqref{constraint:budget}--\eqref{constraint:flowpositive}. 
Since $c(T)\le B$ and $x_v \le 1$, for all $v\in V(T)$, then constraint~\eqref{constraint:budget} is satisfied. Constraint~\eqref{constraint:capacitypositive} is satisfied by construction of $x$.
Finally, there exists a flow function that satisfies constraints~\eqref{constraint:flow},~\eqref{constraint:flowconservation}, and~\eqref{constraint:flowpositive}, given any $x$ such that $1/n\le x_v\leq 1$, for all $v\in V(T)$. In fact: (i) since $T$ is connected, constraints~\eqref{constraint:flow} are satisfied by any flow function that, for all $v\in V(T)$, sends exactly $x_v$ amount of flow from $r$ to $v$; (ii) since $x_w \ge 1/n$ for any $w \in V(T)$, then $nx_w\geq 1$ and  $\sum_{P \in \mathcal{P}_v:w \in P} f^v_{P} \leq \sum_{P \in \mathcal{P}_v}f^v_P = x_v \leq 1$, which implies that constraints~\eqref{constraint:flowconservation} are satisfied by any flow function satisfying constraints~\eqref{constraint:flow}.
\end{proof}
As in the case of additive prize functions, the optimum $OPT_Q$ to the linear program of minimizing $B$ subject to constraints~\eqref{lpUndirected-SubmodularFlow} gives a lower bound to $c(T^*_Q)$, i.e. $OPT_Q \leq c(T^*_Q)$, and the optimum $OPT_B$ to the linear program of maximizing $Q$ subject to constraints~\eqref{lpUndirected-SubmodularFlow}
gives an upper bound to $p(T^*_B)$, $OPT_B \geq p(T^*_B)$. We denote these two linear programs by \UQuotaLP and \UBudgetLP, respectively.

\paragraph{Finding a good tree from a feasible fractional solution to~\eqref{lpUndirected-SubmodularFlow}.}

Here we elaborate the second step. In particular, we find a tree with a good trade-off between prize and cost using the fractional solution from the previous step. We show the following theorem. Recall that $F$ denotes the maximum distance from $r$ to a node in $V$, $F:=\max_{v\in V}\{ dist(r,v) \}$.

\begin{theorem}\label{th:urst:tree-ratio}
Given a feasible solution $x$ to~\eqref{lpUndirected-SubmodularFlow}, then there exists a polynomial time algorithm that computes a tree $T$ such that $r\in V(T)$, and one of the two following conditions holds
(i) $c(T) = O(B\sqrt{n}\ln{n})$ and $p(T)\geq \frac{Q}{2}$ or (ii) $c(T)\leq F$ and $p(T)\geq \frac{Q}{2\sqrt{n}}$.
\end{theorem}
\begin{proof}
Let $x$ be a feasible solution for~\eqref{lpUndirected-SubmodularFlow} and $S \subseteq V$ be the set of vertices $v$ with $x_v>0$ i.e., $S=\{v \in V: x_v>0\}$.
We partition $S$ into two subsets $S_1, S_2 \subseteq S$, where $S_1=\{v \in S| x_v\ge \frac{1}{\sqrt{n}}\}$ and $S_2=\{v \in S| x_v < \frac{1}{\sqrt{n}}\}$, i.e., $S_2=S\setminus S_1$.

We know that $\sum_{v \in S}x_v p_v\geq Q$, which implies that $\sum_{v \in S_1}x_v p_v \ge \frac{Q}{2}$ or  $\sum_{v \in S_2}x_v p_v \ge \frac{Q}{2}$, we then distinguish between two cases.

\begin{enumerate}
    \item $\sum_{v \in S_1}x_v p_v \ge \frac{Q}{2}$. We consider the set of vertices in $S_1 \cup \{r\}$ as the set of terminals in an instance of the fractional node-weighted Steiner tree problem.
    
    In the node-weighted Steiner tree problem, we are given an undirected graph $G''=(V'', E'')$ with nonnegative costs assigned to its nodes and edges and a set of terminals $X \subseteq V''$, and the goal is to find a connected subgraph of $G''$ including $X$ such that its cost (the total cost on its nodes and edges) is minimum. We term this problem \textbf{NWST}. Klein and Ravi~\cite{klein1995nearly} gave an $O(\log n)$-approximation algorithm for \textbf{NWST}. In the fractional relaxation of \textbf{NWST} we need to assign capacities to nodes and edges in such a way that the capacity of any \emph{cut} (i.e. a subset of nodes and edges) that separates any terminal from $r$ is at least $1$ and the sum of edge and vertex capacities multiplied by their cost is minimized.
    We term this problem \textbf{FNWST}. Let $SOL_{FNWST}$ be an optimal solution for \textbf{FNWST}. Guha, Moss, Naor, and Schieber~\cite{guha1999efficient} showed that we can generalize Klein and Ravi's technique~\cite{klein1995nearly} and provide a solution $T_{NWST}$ to \textbf{NWST} (i.e. a tree spanning all terminals) obtained from $SOL_{FNWST}$ such that its cost is $O(\log n)$ times the cost of $SOL_{FNWST}$, i.e., $c(T_{NWST})=O(\log n)c(SOL_{FNWST})$.
    
    Starting from $G$ and $x$, we define an instance $I_{FNWST}$ of \textbf{FNWST} in which the set of terminals is $S_1 \cup \{r\}$, the node costs are defined by $c$ and the edge costs are zero. Let $OPT_{FNWST}$ be the optimum for $I_{FNWST}$. The feasible solution $x$ for~\eqref{lpUndirected-SubmodularFlow} induces a solution to  $I_{FNWST}$ with cost at most $B \sqrt{n}$. This is because the cost of $x$ is at most $B$ and $x_v\geq 1/\sqrt{n}$ for each vertex $v\in S_1$. So, if we multiply $x_v$ by a factor at most $\sqrt{n}$, for each vertex $v\in S_1$, we obtain a solution to $I_{FNWST}$ which is feasible, since the capacity of any cut that separates any terminal from $r$ is at least $1$, and costs at most $B \sqrt{n}$. Therefore, $OPT_{FNWST} \leq B \sqrt{n}$.
    Now, by using the $O(\log n)$-approximation algorithm of~\cite{guha1999efficient}, we compute a tree $T$ with cost $c(T)=O(B \sqrt{n}\log n)$ that spans all the vertices in $S_1$ and hence, by monotonicity of $p$ and by constraint~\eqref{constraint:submodular}, has prize $p(T)\ge p(S_1) \ge \sum_{v \in S_1}x_v p_v \ge\frac{Q}{2}$.
    
    \item $\sum_{v \in S_2}x_v p_v \ge \frac{Q}{2}$. In this case, we have 
    
    \begin{align}\label{eqCase1LowerBoundPrizeUndirected}
        \sum_{v \in S_2}p_v\ge \sqrt{n}\times \sum_{v \in S_2}x_v p_v \ge \frac{\sqrt{n} Q}{2},
    \end{align}
    where the first inequality holds since $x_v < \frac{1}{\sqrt{n}}$, for any $v \in S_2$, and the second inequality holds by the case assumption. 
    
    Inequality~\eqref{eqCase1LowerBoundPrizeUndirected} implies that the vertex $v \in S_2$ with the highest prize has $p(\{v\}) \ge \frac{Q}{2 \sqrt{n}}$, since $|S_2| \le n$. Then we find a shortest path $P$ from $r$ to $v$ and return the tree $T=P$ as the solution in this case. Note that since $p$ is monotone non-decreasing, then $p(T) \ge \frac{Q}{2 \sqrt{n}}$, moreover, $c(T)\le F$ as $dist(r,u)\le F$.\qedhere
\end{enumerate}
\end{proof}

\paragraph{Approximation algorithm for \UQuotaName.}
By similar argument to that used for the approximation algorithm for \DSteinerT and \DQuotaAdditiveName, we can assume that each vertex $v$ has a distance no more than $(1+\epsilon)c(T^*_Q)$ from $r$ in $G$, that is $F\le (1+\epsilon)c(T^*_Q)$, for any $\epsilon >0$.
Therefore, using an optimal solution $x$ to \UQuotaLP as input to the algorithm in Theorem~\ref{th:urst:tree-ratio}, we obtain the following theorem.
\thUQuotaTree*

\paragraph{Approximation algorithm for \UBudgetName.}
Let $x$ be an optimal solution for~\UBudgetLP and $T$ be the tree computed by the algorithm in Theorem~\ref{th:urst:tree-ratio} when $x$ is used as input. Tree $T$ can violate the budget constraint if used as a solution for $I_B$. However, the budget violation and  the prize-to-cost ratio of $T$ are bounded, i.e. $c(T)=O(B \sqrt{n}\log n)$ and $\gamma=\frac{p(T)}{c(T)}=\Omega(\frac{p(T^*_B)}{B\sqrt{n}\log{n}})$. In the following, we show how to trim $T$ to have a budget violation of $1+\eps$ and an approximation ratio of $O(\frac{\sqrt{n}\log{n}}{\epsilon^3})$, for any $\epsilon\in (0,1]$. We again assume w.l.o.g. that $F\leq B$, i.e. graph $G$ is $B$-proper for $r$, and we will use two results from~\cite{d2022budgeted}.

Let $D$ be a directed graph where each node is associated with a cost, and the prize is defined by a monotone submodular function on the subsets of nodes. The following lemma introduced a trimming process that takes as input an out-tree of $D$, and returns another out-tree of $D$ which has a smaller cost but preserves the same prize-to-cost ratio (up to a bounded multiplicative factor). 

\begin{lemma}[Lemma 4.2 in~\cite{d2022budgeted}]\label{trimmingProcess}
Let $D= (V, A)$ be a $B$-proper graph for a node $r$. Let $\Tilde T$ be an out-tree of $D$ rooted at $r$ with  prize-to-cost ratio $\gamma=\frac{p(\Tilde T)}{c(\Tilde T)}$, where $p$ is a monotone submodular function. Suppose that $\epsilon B/2\le c(\Tilde T) \le hB$, where $h \in (1, n]$ and $\epsilon \in (0, 1]$. One can find an out-subtree $\hat{T}$ rooted at $r$ with the prize-to-cost ratio at least $\frac{\epsilon^2 \gamma}{32h}$ such that $\epsilon B/2 \le c(\hat{T}) \le (1+\epsilon)B$.
\end{lemma}

This trimming process returns a tree with a good prize as long as the factor $h$ by which the budget is violated in $\Tilde T$, is small. However, it returns a tree $\hat T$ with a low prize when $h$ is large. 
In the following, we introduce an improvement of the above trimming process, which might be used for different problems.

\begin{lemma}\label{lmNewtrimmingProcess}
Let $D= (V, A)$ be a $B$-proper graph for a node $r$. Let $T$ be an out-tree of $D$ rooted at $r$ with  prize-to-cost ratio $\gamma=\frac{p(T)}{c(T)}$, where $p$ is a monotone submodular function. Suppose that $\epsilon B/2\le c(T) \le hB$, where $h \in (1, n]$ and $\epsilon \in (0, 1]$. One can find in polynomial time an out-subtree $\hat{T}$ rooted at $r$ such that one of the two following conditions holds: the prize-to-cost ratio of $\hat{T}$ is at least $\frac{\epsilon^2 \gamma}{640}$ and $\epsilon B/2 \le c(\hat{T}) \le (1+\epsilon)B$; $p(\hat{T}) \ge p(T)/5h$ and $c(\hat{T}) \le B$.
\end{lemma}

Before proving Lemma~\ref{lmNewtrimmingProcess}, we need the following lemma.

\begin{lemma}[Lemma 4.3 in~\cite{d2022budgeted}] \label{lmClaimKuoExtension}
For any out-tree $\Tilde T=(V, A)$ rooted at $r$ with cost $c(\Tilde T)$ and any $m\leq c(\Tilde T)$, there exist $N \le 5\lfloor \frac{c(\Tilde T)}{m}\rfloor$ out-subtrees $T^i=(V^i, A^i)$ of $\Tilde T$, for $i \in [N]$, where $V^i \subseteq V$, $A^i = (V^i \times V^i) \cap A$, $c(V^i) \le m+c(r_i)$, $r_i$ is the root of $T^i$, and $\bigcup_{i=1}^{N} V^i=V$.
\end{lemma}

\begin{proof}[Proof of Lemma~\ref{lmNewtrimmingProcess}]
Let $c(T) = h'B$, where $h' \leq h$. We first decompose $T$ into $5\lfloor h'\rfloor$ out-subtrees by applying Lemma~\ref{lmClaimKuoExtension} with $m=B$; each computed out-subtree $T^i \subseteq T$, $i\in [5\lfloor h'\rfloor ]$, costs at most $B+c(r^i)$, where $r^i$ is the root of $T^i$. Note that the proof of the lemma provides a polynomial time algorithm to compute such a decomposition (see~\cite{d2022budgeted}).

Then, we select the out-subtree, say $T^z$, maximizing the prize. By the submodularity of $p$, we have $p(T^z)\ge \frac{p(T)}{5h'}$ . If $r^z=r$, then we define $T'=T^{z}$. Otherwise, we compute a shortest path $P$ from $r$ to $r^z$ and define $T'$ as the union of $T^{z}$ and $P$. Since the obtained graph might not be an out-tree, we remove the possible edges incoming the nodes in $V(T^{z})\cap V(P)$ that belong only to $T^{z}$. We have $c(T') \le 2B$ as $dist(r, r^z) \le B$ and $c(T^{z}\setminus \{r^z\}) \le B$. By monotonicity of $p$, we have also that $p(T')\ge p(T^z) \ge \frac{p(T)}{5h'}$, and hence $T'$ has a prize-to-cost ratio
\[
\gamma' = \frac{p(T')}{c(T')}\geq\frac{p(T)}{10h'B}= \frac{p(T)}{10c(T)}=\frac{\gamma}{10}.
\]

If $c(T')\leq \epsilon B/2 \leq B$, then we output $\hat{T} = T'$, which satisfies the second condition of the lemma statement.
Otherwise, we can use Lemma~\ref{trimmingProcess} with $\Tilde T=T'$ and $h =2$ and compute another tree $\hat T$ with cost  $\epsilon B/2 \leq c(\hat T)\le (1+\epsilon)B$  and prize-to-cost ratio equal to
\[
\frac{p(\hat{T})}{c(\hat{T})}\geq  \frac{\epsilon^2 \gamma'}{32h}\ge \frac{\epsilon^2 \gamma'}{64}
\geq \frac{\epsilon^2 \gamma}{640}.
\]
The proof is complete.
\end{proof}

Now we show the main theorem of this section.
\MainUBudgetTheorem*
\begin{proof}

We first compute an optimal solution $x$ for \UBudgetLP and then we apply the algorithm in Theorem~\ref{th:urst:tree-ratio} and compute a tree $T$. 
Since $x$ is a feasible solution to~\eqref{lpUndirected-SubmodularFlow} when $Q=OPT_B \geq p(T^*_B)$, then one of the two following conditions holds: (i) $c(T) = O(B\sqrt{n}\ln{n})$ and $p(T)\geq \frac{p(T^*_B)}{2}$ or (ii) $c(T)\leq B$ and $p(T)\geq \frac{p(T^*_B)}{2\sqrt{n}}$. In case (ii) the theorem holds.

In case (i), we have that $c(T) = hB$, for some $h = O(\sqrt{n}\ln{n})$, therefore we apply Lemma~\ref{lmNewtrimmingProcess} and obtain a new tree $\hat{T}$ such that one of the two following conditions holds:  $\frac{p(\hat{T})}{c(\hat{T})} \geq \frac{\epsilon^2 p(T)}{640 c(T)}$ and $\epsilon B/2 \le c(\hat{T}) \le (1+\epsilon)B$; $p(\hat{T}) \ge p(T)/5h$ and $c(\hat{T}) \le B$. In the first case, we have 
\[
p(\hat{T})\geq \frac{\epsilon^2p(T)}{640h B} c(\hat{T})\geq
\frac{\epsilon^3p(T)}{1280h} =
\Omega\left(\frac{\epsilon^3p(T)}{\sqrt{n}\ln{n}}\right).
\]
In the second case, we have $h=O(\sqrt{n}\ln{n})$ that implies $p(\hat{T})=\Omega\left(\frac{p(T)}{\sqrt{n}\ln{n}}\right)$. Observing that $p(T)\geq \frac{p(T^*_B)}{2}$ concludes the proof.
\end{proof}

\section{Conclusion}

We obtained very simple polynomial time approximation algorithms for some budgeted and quota variants of node-weighted Steiner tree problems on two scenarios: (i) directed graphs with additive prizes, and (ii) undirected graphs with submodular prizes. The key insights behind our algorithms for the first scenario were to carefully select a subset of vertices as terminals in a fractional solution returned by the standard flow-based LPs and use some flow properties to find a small hitting set through which all the chosen terminals can be reached from the root vertex. The key idea of our algorithms for the second scenario were to use the submodular flow problem and introduce new LPs for our problems.

To the best of our knowledge, our techniques yield the first polynomial time (bicriteria) approximation algorithms for these problems (except \DSteinerT) in terms of the number of vertices $n$. Furthermore, we believe that our introduced LPs can be utilized for some other Steiner problems in which, for example, other constraints can be added to~\eqref{lpDBudgetQuotaAdditive} and~\eqref{lpUndirected-SubmodularFlow}. 

A natural open question asks to improve the approximation guarantees or prove that the current
guarantees are the best possible using the flow-based LPs. By the result of Bateni, Hajiaghay and Liaghat~\cite{bateni2018improved}, we know that the integrality gap of the flow-based LP for \DBudgetAdditiveName is infinite. This implies that, using this LP, we should work on improving the approximation factors for the budgeted problems while violating the budget constraint. Also, Li and Laekhanukit~\cite{li2022polynomial} showed that the integrality gap of the flow-based LP for \DSteinerT is polynomial in the number of vertices $n$. This means that using this LP, one needs to work on the possibility of achieving an approximation algorithm with the factor $O(n^{\epsilon})$ for \DSteinerT, where $0 <\epsilon < 1/2$. Another interesting future work would be the possibility of extending our techniques for second scenario into directed graphs.
\bibliographystyle{alpha}
\bibliography{biblio.bib}

\newpage
\appendix

\section*{Appendix}

\section{An Equivalent formulation of flow constraints}\label{apx:lp}
All our set of constraints and linear programs have an exponential number of variables. However, they can be solved in polynomial time as we only need to find, independently for any $v \in V\setminus \{r\}$, a flow from $r$ to $v$ of value $x_v$ that does not exceed the capacity $x_w$ (and $nx_w$ in~\eqref{lpUndirected-SubmodularFlow}), for each vertex $w \in V\setminus\{v\}$.
Indeed, taking the example of~\eqref{lpDBudgetQuotaAdditive}, the flow variables appear only in constraints~\eqref{lpDBudgetQuotaAdditive:overrallflow} and~\eqref{lpDBudgetQuotaAdditive:capacity}, while the quota and budget constraints only depend on the capacity variables. 
Therefore, we can replace the flow variables and constraints \eqref{lpDBudgetQuotaAdditive:overrallflow} and~\eqref{lpDBudgetQuotaAdditive:capacity}, with an alternative formulation of flow variables and constraints in such a way that for any assignment $x$ of capacity variables, there exists a feasible assignment of flow variables if and only if there exists a feasible assignment of the alternative flow variables.

The two following sets of constraints, where  $x$ gives capacity values, are equivalent in this sense.
\begin{align}\label{const1}
& \tag{Const1}\\
                  \sum_{P \in \mathcal{P}_v}f^v_P &= x_v,&& \forall v \in V\setminus \{r\}\\
                 \sum_{P \in \mathcal{P}_v:w \in P} f^v_{P}&\le x_w,&& \forall v\in V\setminus \{r\} \text{ and }\forall w \in V\setminus \{ v\}\\
                 0 \le f^v_P &\le 1, &&\forall v\in V\setminus \{r\}, P \in \mathcal{P}_v\notag
\end{align}

\begin{align}\label{const2}
& \tag{Const2}\\
                  \sum_{w\in V} f^v_{wv} &= x_v,&& \forall v \in V\setminus \{r\}\label{const2:overallflow}\\
                  \sum_{u\in V} f^v_{wu} &\le x_w,&& \forall v\in V\setminus \{r\} \text{ and }\forall w \in V\setminus \{ v\}\label{const2:capacity}\\
                 \sum_{u \in V} f^v_{wu}&=\sum_{u \in V} f^v_{uw}&& \forall v\in V\setminus \{r\} \text{ and }\forall w \in V\setminus \{r, v\}\label{const2:conservation}\\\
                 0 \le &f^v_{wu} \le 1, &&  \forall v\in V\setminus \{r\} \text{ and }\forall w,u\in V\notag
\end{align}
In both formulations, the values of $x$ are fixed, while $f_P^v$, for each $v\in V\setminus\{r\}$, and $f^v_{wu}$, for each $v\in V\setminus \{r\}$ and $w,u \in V\setminus \{v\}$, are the flow variables for~\eqref{const1} and~\eqref{const2}, respectively.
In~\eqref{const2} for any $v \in V \setminus \{r\}$, $r$ has to send $x_v$ units of commodity $v$ to every vertex $v$ and $f^v_{wu}$ is the flow of commodity $v$ on the directed edge $(w, u)$.  
The constraints in~\eqref{const2} are as follows. Constraint~\eqref{const2:overallflow} ensures that the amount of flow entering each vertex $v \in V\setminus \{r\}$ should be equal to $x_v$. Constraints~\eqref{const2:capacity} and~\eqref{const2:conservation} formulate the standard flow constraints encoding of the connectivity constraint in which in a flow from $r$ to $v$, $x_u$ is the capacity of each vertex $u \in V \setminus \{v\}$ and $f^v_{wu}$ is the flow of commodity $v$ on the directed edge $(w,u)$.

It is easy to see that, given an assignment of capacity variable $x$, there exists an assignment of variables $f^v_P$, for each $v\in V\setminus\{r\}$ and $P \in \mathcal{P}_v$, that satisfies constraints~\eqref{const1} if and only if there exists an  assignment of variables $f^v_{wu}$, for each $v\in V\setminus \{r\}$ and $w,u \in V\setminus \{v\}$, that satisfies constraints~\eqref{const2}. In fact, both conditions are satisfied if an only if it is possible for each $v\in V\setminus \{r\}$, to send $x_v$ units of flow from $r$ to $v$ satisfying the node capacities defined by $x$. 

In our approximation algorithms, we will use only the capacity variables, which we can compute by replacing~\eqref{const1} with~\eqref{const2} in the respective linear program and hence solving linear programs with a polynomial number of variables.

\section{Proof of Claim~\ref{clHittingSet}}\label{apx:claim}

For each element $i\in V'$, we define a counter $c_i$ of the number of sets which $i$ belongs to, $c_i:=|\{j~:~ i\in X'_j\}|$. We initialize $X'$ to $\emptyset$ and iterate the following greedy steps until $X'$ hits all the subsets of $\Sigma$:
(1) select the element $i$ that maximizes $c_i$; (2) add $i$ to $X'$; (3) update all counters of elements that belong to a set which also $i$ belongs to.

The above algorithm runs for at most $N$ iterations since at least a subset is covered in each iteration. Moreover, each iteration requires polynomial time in $N$ and $M$.

For $k\geq 0$, let $N_k$ be the number of sets not covered by $X'$ after $k$ iterations of the above algorithm.
We have $N_0=N$ and $N_{|X'|} =0$.
Let $i$ be the element selected at iteration $k\geq 1$.
At the beginning of iteration $k$, before adding $i$ to $X'$, we have that 
\[
\sum_{\ell:c_\ell>0} c_\ell \geq R \cdot N_{k-1},
\]
and since $|\{\ell:c_\ell>0\}| \leq M-k+1$, by an averaging argument, we must have $c_i\geq \frac{R N_{k-1}}{M-k+1}$.
It follows that:

\[
N_k =N_{k-1}-c_i\leq  \left( 1 - \frac{R}{M-k+1}\right) N_{k-1}\leq  N_0 \prod_{\ell=0}^{k-1} \left( 1 - \frac{R}{M-\ell}\right) = N \prod_{\ell=0}^{k-1} \left( 1 - \frac{R}{M-\ell}\right) <  N \left( 1 - \frac{R}{M}\right)^{k}\leq N e^{-Rk/M}~,
\]
where the second inequality is due to $k-1$ recursions on $N_{k-1}$ and the last one is due to $1-x\leq e^{-x}$, for any $x\geq 0$. For $k=\frac{M}{R}\ln{N}$ we have $N_k< 1$, which means $N_k=0$ and hence $|X'|\leq \frac{M}{R}\ln{N}$.


\end{document}